%% file: paper.tex
\documentclass[a4paper,runningheads]{llncs}

\usepackage[T1]{fontenc}

\usepackage[table]{xcolor}
\newcommand{\CC}[1]{\cellcolor{gray!#1}}

\usepackage{pgfplots}
\usepackage{amsmath,amssymb}
\usepackage{mathtools}
\usepackage{todonotes}
\usepackage[basic,small,classfont=bold]{complexity}
\usepackage[capitalise]{cleveref}
\usepackage{comment}
\usepackage{proof}
\usepackage{graphicx}

\usepackage[linesnumbered,noend]{algorithm2e}
\RestyleAlgo{ruled}
\SetKwComment{Comment}{/* }{ */}
\SetKwInOut{Input}{Inputs}
\SetKwInOut{Output}{Output}

\usepackage{array}
\usepackage{makecell}
\usepackage{multirow,bigdelim}
\usepackage[caption=false]{subfig}

\usepackage{tikz}
\usetikzlibrary{shapes,shapes.geometric,arrows,fit,calc,positioning,automata,chains,matrix.skeleton, arrows.meta}
\tikzset{nature/.style={draw,rectangle}}
\tikzset{>={Stealth[scale=1.2]}} 

\newcommand{\until}{\,\mathsf{U}\,}
\newcommand{\MVal}{\mathcal{M}[\val]}

\newcommand{\MDPValSigma}{\mathcal{M}_{\val}^{\sigma}}
\newcommand{\MDPValSigmaN}{\mathcal{N}_{\val}^{\sigma}}
\newcommand{\MDPValprimeSigma}{\mathcal{M}_{\val'}^{\sigma}}
\newcommand{\MDPValprimeSigmaN}{\mathcal{N}_{\val'}^{\sigma}}
\newcommand{\GraphPresVals}{\mathsf{Val}^{\mathrm{gp}}_{\mathcal{M}}}
\newcommand{\GraphPresValsN}{\mathsf{Val}^{\mathrm{gp}}_{\mathcal{N}}}
\newcommand{\ReachProb}[3]{\mathbb{P}^{#1}_\mathcal{#2}[\Diamond #3]}
\newcommand{\ReachProbM}[1]{\mathbb{P}^{#1}_{\mathcal{M}^\val}[\Diamond \fin]}
\newcommand{\ReachProbu}[4]{\mathbb{P}^{#1}_\mathcal{#2}[#3\until #4]}

\newcommand{\mmec}{\mathcal{M}_{/\textnormal{MEC}}}

\newcommand{\RewardNu}{\mathrm{Rew}^{\nu}}
\newcommand{\Reward}[2]{\mathrm{Rew}^{#1}_{#2}}
\newcommand{\RewardVal}{\mathrm{Rew}^\sigma_{\val}}
\newcommand{\RewardValM}{\mathrm{Rew}^\sigma_{\mathcal{M}[\val]}}
\newcommand{\RewardValN}{\mathrm{Rew}^\sigma_{\mathcal{N}[\val]}}
\newcommand{\RewardValprimeN}{\mathrm{Rew}^\sigma_{\mathcal{N}[\val']}}
\newcommand{\RewardOptVal}{\mathrm{Rew}^*_{\val}}
\newcommand{\RewardOptValM}{\mathrm{Rew}^*_{\mathcal{M}[\val]}}
\newcommand{\RewardOptValN}{\mathrm{Rew}^*_{\mathcal{N}[\val]}}
\newcommand{\RewardOptValprimeN}{\mathrm{Rew}^*_{\mathcal{N}[\val']}}

\newcommand{\fin}{\mathit{fin}}
\newcommand{\fail}{\mathit{fail}}
\newcommand{\val}{\mathsf{val}}
\newcommand{\valprime}{\mathsf{val'}}

\newclass{\ETR}{ETR}
\newcommand{\textETR}{$\ETR$}
\newcommand{\textCoETR}{$\co\ETR$}

\newcommand{\textCoNP}{$\co\NP$}
\newcommand{\Set}[1]{\{#1\}}
\newcommand{\set}[1]{\Set{#1}}
\newcommand{\multiset}[1]{ \{\!\{ #1 \}\!\} }
\newcommand{\mecclass}[1]{ [\![ #1 ]\!]_{\mathrm{MEC}} }

\newcommand{\say}[1]{``#1''}

\crefname{property}{Property}{Properties}

\begin{document}

\title{Graph-Based Reductions for Parametric and Weighted MDPs}
%
%
\author{Kasper Engelen
\and
Guillermo A. P\'{e}rez
\and
Shrisha Rao
}

\authorrunning{K. Engelen et al.}

\institute{University of Antwerp -- Flanders Make, Belgium\\ \email{\{kasper.engelen, guillermo.perez, shrisha.rao\}@uantwerpen.be}}

\maketitle

\begin{abstract}
  We study the complexity of reductions for weighted reachability in
  parametric Markov decision processes. That is, we say a state $p$ is never
  worse than $q$ if for all valuations of the polynomial indeterminates it is
  the case that the maximal expected weight that can be reached from $p$ is
  greater than the same value from $q$. In terms of computational complexity,
  we establish that determining whether $p$ is never worse than $q$ is
  \textCoETR-complete. On the positive side, we give a polynomial-time
  algorithm to compute the equivalence classes of the order we study for
  Markov chains. Additionally, we describe and implement two inference rules
  to under-approximate the never-worse relation and empirically show that it
  can be used as an efficient preprocessing step for the analysis of large
  Markov decision processes.
\end{abstract}

\section{Introduction}


Markov decision processes (MDPs, for short) are useful mathematical models to
capture the behaviour of systems involving some random components and discrete
(non-deterministic) choices. In the field of verification, they are studied as
formal models of randomised algorithms, protocols,
etc.~\cite{baier2008,DBLP:reference/mc/2018} In artificial intelligence, MDPs
provide the theoretical foundations upon which reinforcement learning
algorithms are based~\cite{DBLP:books/aw/RN2020,DBLP:books/lib/SuttonB98}.


From the verification side, efficient tools have been implemented for the analysis
of MDPs. These include, for instance, PRISM~\cite{DBLP:conf/cav/KwiatkowskaNP11},
Storm~\cite{DBLP:journals/sttt/HenselJKQV22}, and 
Modest~\cite{DBLP:conf/tacas/HartmannsH14}. While those
tools can handle ever larger and complexer MDPs, as confirmed by the most
recent Comparison of Tools for the Analysis of Quantitative Formal
Models~\cite{DBLP:conf/isola/BuddeHKKPQTZ20},
they still suffer from the state-explosion problem and the worst-case
exponential complexity of value iteration. While most tools can check the MDP
against a rich class of specifications (linear temporal logic, probabilistic
computation-tree logic, etc.), they often reduce the task to reachability
analysis and solve the latter using some version of value
iteration~\cite{DBLP:journals/corr/abs-2301-10197}.


In this work we focus on reduction techniques to improve the runtime of MDP
analysis algorithms as implemented in the aforementioned verification tools.
We are particularly interested in reductions based on the graph underlying the
MDP which do not make use of the concrete probability distributions. This,
because such reductions can also be applied to partially known MDPs that arise
in the context of reinforcement learning~\cite{bharadwaj2017}. Several graph
and automata-theoretic reduction techniques with these properties already
exist. They focus on the computation of states with specific properties, for
instance: extremal-value states (i.e. states whose probability of reaching the
target set of states is $0$ or $1$), maximal end components (i.e. sets of
states which can almost-surely reach each other), and essential
states~\cite{DBLP:conf/papm/DArgenioJJL01,DBLP:conf/qest/CiesinskiBGK08} (i.e.
sets of states which form a graph-theoretic separator between some given state
and the target set), to name a few. In~\cite{NWR_paper}, Le Roux and P\'erez
have introduced a partial order among states, called the \emph{never-worse
relation} (NWR, for short), that subsumes all of the notions listed above.
Their main results in that work were establishing that the NWR the natural
decision problem of comparing two states is \coNP-complete, and giving a few
inference rules to under-approximate the full NWR.

\paragraph{Our contribution.}
In this paper, we have extended the NWR to weighted reachability in parametric MDPs.
Formally, we have assigned weights to the target set of states and allow for transitions to be labelled with polynomials instead of just concrete
probability values. Then, we say a state $p$ is never worse than $q$ if for
all valuations of the indeterminates, it is the case that the maximal expected
weight that can be reached from $p$ is at least the same value from $q$.
We have shown that determining whether $p$ is NWR than $q$ is \textCoETR-complete, so it
is polynomial-time interreducible with determining the truth value of a
statement in the existential theory of the reals. Along the way, we prove that
the NWR for weighted reachability reduces in polynomial time to the NWR for
(Boolean) reachability. Regarding theoretical results, we further establish
that the equivalence classes of the NWR can be computed in polynomial time for a special class of
parametric Markov chains, while this is just as hard as the NWR decision
problem for parametric MDPs. Finally, we have significantly improved (and corrected small errors in) the inference
rules from Le Roux and P\'erez and provided a concrete algorithm to use them in
order to reduce the size of a given MDP. We have implemented our algorithm and
evaluated its performance on a number of benchmarks from the Quantitative
Verification Benchmark Set~\cite{DBLP:conf/tacas/HartmannsKPQR19}.

\section{Preliminaries}

For a directed graph $\mathcal{G} = (V, E)$ and a vertex $u \in V$, we write
$uE = \Set{ v \in V \mid (u,v) \in E}$ to denote the set of \emph{(immediate)
successors} of $u$. A \emph{sink} is a vertex $u$ such that $uE = \emptyset$.

For a finite alphabet $\Sigma$, $\Sigma^*$ is the set of all finite words over
$\Sigma$ including the empty word $\varepsilon$, and $\Sigma^+$ is the set of
all non-empty words.

Let $X = \Set{x_1, \dots, x_k}$ be a finite set of variables and
$p(x_1,\dots,x_k)$ be a polynomial with rational coefficients on $X$.
For a valuation $\val : X \to \mathbb{R}$ of the variables, we write $p[\val]$
for the image of $p$ given $\val(X)$, i.e.
$p(\val(x_1),\dots,\val(x_k))$. We write $p \equiv 0$ to denote the fact that
$p$ is syntactically equal to $0$ and $p = 0$ to denote that its image is $0$ for all
valuations. Finally, we write $\mathbb{Q}[X]$ to denote the set of all
polynomials with rational coefficients on $X$. 

\subsection{Stochastic models}


A distribution on a finite set $S$ is a function $f : S
\rightarrow \mathbb{R}_{\geq 0}$, with $\sum_{s \in S} f(s) = 1$. 


\begin{definition}[Markov chains]
  A (weighted) Markov chain is a tuple $\mathcal{C} = (S, \mu, T, \rho)$ with
  a finite set of states $S$, a set of target states $T\subseteq S$, a
  probabilistic transition function $\mu : (S\setminus T) \times S \rightarrow
  \mathbb{R}_{\geq 0}$ with $ \forall s \in S \setminus T : \sum_{s' \in S}
  \mu(s,s') = 1$ (i.e. $\mu$ maps non-target states onto probability
  distributions over states), and a weight function
  $\rho:T\rightarrow\mathbb{Q}$ that assigns weights to the target states.
\end{definition}
Note that all states in $T$ are sinks: $\mu(t,s)$ is
undefined for all $t \in T,s \in S$.

A \emph{run} of $\mathcal{C}$ is a finite non-empty word $\pi = s_0 \dots s_n$
over $S$ such that $0 < \mu(s_i,s_{i+1})$ for all $0 \leq i < n$. The run $\pi$
\emph{reaches} $s' \in S$ if $s' = s_n$. The probability associated with a run
is defined as $\mathbb{P}_{\mu}(\pi) = \prod_{0 \leq i < n} \mu(s_i,
s_{i+1})$.

We also define the probability of reaching a set of states.
Below, we write $\multiset{x_1, \dots, x_n}$ to denote the \emph{multiset} consisting of the elements $x_1, \dots, x_n$.

\begin{definition}[Reachability]
Given a Markov chain $\mathcal{C}=(S,\mu,T,\rho)$, an initial state $s_0\in S$
and a set of states $B\subseteq S$, we denote by
$\ReachProb{s_0}{\mathcal{C}}{B}$ the probability of reaching $B$
from $s_0$. If $s_0 \in B$ then $\mathbb{P}^{s_0}_\mathcal{C}[\Diamond B] = 1$,
otherwise:
  \[
    \mathbb{P}^{s_0}_\mathcal{C}[\Diamond B] = \sum \multiset{
    \mathbb{P}_{\mu}(\pi) \mid \pi = s_0 \dots s_n \in (S \setminus
    B)^+ B }.
  \]
For brevity, when $B=\Set{a}$, we write $\ReachProb{s_0}{\mathcal{C}}{a}$ instead of $\ReachProb{s_0}{\mathcal{C}}{\Set{a}}$.
\end{definition}
Now, we use the weights of targets states to define a weighted reachability value.

\begin{definition}[Value of a state]
Given a Markov chain $\mathcal{C}=(S,\mu,T,\rho)$ and a state $s_0\in
S$, the (expected reward) value of $s_0$ is
\(
  \Reward{s_0}{\mathcal{C}}=\sum_{t\in T} \ReachProb{s_0}{\mathcal{C}}{t}\cdot\rho(t).
\)
\end{definition}
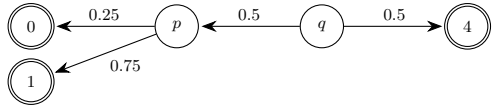
\begin{figure}[b]
    \centering
        
    \begin{tikzpicture}[shorten >=1pt,auto,node distance=1.9 cm, scale = 0.7, transform shape]

        \node[state,accepting](0){$0$};
        \node[state](p)[right=of 0]{$p$};
        \node[state,accepting](1)[below=0.2cm of 0]{$1$};
        \node[state](q)[right=of p]{$q$};
        \node[state,accepting](4)[right=of q]{$4$};
                
        \path[->,auto] 
          (p)   edge node[swap] {$0.25$} (0)
          (p)   edge node {$0.75$} (1)
          (q)   edge node[swap] {$0.5$} (p)
              edge node {$0.5$} (4)
        ;
    \end{tikzpicture}
    \caption{Example of a Markov chain $\mathcal{C}$ with
      $\Reward{p}{\mathcal{C}}=0.75$ and
      $\Reward{q}{\mathcal{C}}=2.375$. Circles depict states; arrows, transitions; double
  circles, target states with their integer labels being their weights.}
    \label{fig:MC_Rew}
\end{figure}

Finally, we recall a definition of parametric Markov decision processes.

\begin{definition}[Markov decision processes]
\label{def:W_pMDP}
A weighted and parametric Markov decision process (wpMDP) is a tuple
$\mathcal{M} = (S, A, X, \delta, T, \rho)$ with finite sets $S$ of states; $A$
of actions; $X$ of parameters; and $T \subseteq S$ of target states. $\delta : (S\setminus T)\times A \times S \to \mathbb{Q}[X]$ is a probabilistic transition function that maps transitions onto polynomials
and $\rho : T \rightarrow \mathbb{Q}$ is a weight function.
\end{definition}
We assume, without loss of generality, that there is an extra state $\fail\in T$, with $\rho(\fail) = 0$. Intuitively, this is a ``bad'' target state that we want to avoid. Throughout the paper, we sometimes only specify partial probabilistic transition functions. This slight abuse of notation is no loss of generality as (for all of our purposes) we can have all unspecified transitions lead to $\fail$ with probability $1$.

In this paper, we will work with the following subclasses of wpMDPs.
\begin{description}
    \item[pMDPs] or (non-weighted) parametric MDPs are the subclass of wpMDPs
      which have only two states in $T$, namely, $T = \Set{\fin,\fail}$ where
      $\rho(\fin)=1$ and $\rho(\fail)=0$. For pMDPs, we omit $\rho$ from the
      tuple defining them since we already know the weights of the targets.
    \item[w\~pMDPs] or weighted \emph{trivially parametric} MDPs have, for
      each transition, a unique variable as polynomial. That is, 
      the
      probabilistic transition function is
      such that 
      $\delta(s,a,s')\equiv0$ or $\delta(s,a,s')=x_{s,a,s'}\in X$, for all
      $(s,a,s') \in (S \setminus T) \times A \times S$.
      Since all parameters are trivial in such
      wpMDPs, we omit their use. Instead, an w\~pMDPs is a tuple $(S, A,
      \Delta, T, \rho)$ where $\Delta \subseteq (S\setminus T) \times A
      \times S$ represents all transitions that do not have probability $0$.
    \item[\~pMDPs] or (non-weighted) trivially parametric MDPs are both
      non-weighted and trivially parametric. We thus omit
      $\rho$ and $X$ from their tuple representation $(S,A,\Delta,T)$ where
      $\Delta$ is as defined for w\~pMDPs.
\end{description}

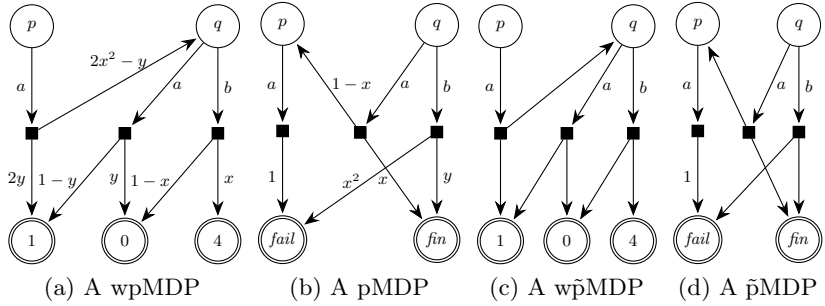
\begin{figure}[htbp]
\centering
\subfloat[A wpMDP\label{fig:wpMDPs-a}]{
\begin{tikzpicture}[shorten >=1pt,auto,node distance=.9 cm, scale = 0.7, transform shape]
        \tikzstyle{action} = [fill=black, shape=rectangle, draw]

        \node[state,accepting](1){$1$};
        \node[action](ap)[above=1.5cm of 1]{};
        \node[state](p)[above=1.5cm of ap]{$p$};
        \node[state,accepting](0)[right=of 1]{$0$};
        \node[action](aq)[above=1.5cm of 0]{};
        \node[state,accepting](4)[right=of 0]{$4$};
        \node[action](bq)[above=1.5cm of 4]{};
        \node[state](q)[above=1.5cm of bq]{$q$};
                
        \path[->] 
        (p)   edge [above] node [left] {$a$} (ap)
        (q)   edge [above] node [right] {$a$} (aq)
              edge [above] node [right] {$b$} (bq)
        (ap)  edge [above] node [left] {$2y$} (1)
              edge [above] node [above=.15cm] {$2x^2-y$} (q)
        (aq)  edge [above] node [left] {$1-y$} (1)
              edge [above] node [left] {$y$} (0)
        (bq)  edge [above] node [left] {$1-x$} (0)
              edge [above] node [right] {$x$} (4)
        ;
    \end{tikzpicture}
}
\subfloat[A pMDP] {
\begin{tikzpicture}[shorten >=1pt,auto,node distance=.9 cm, scale = 0.7, transform shape]
        \tikzstyle{action} = [fill=black, shape=rectangle, draw]

        \node[state,accepting](fail){$\fail$};
        \node[action](ap)[above=1.5cm of fail]{};
        \node[state](p)[above=1.5cm of ap]{$p$};
        \node[](0)[right=of fail]{};
        \node[action](aq)[above=1.8cm of 0]{};
        \node[state,accepting](fin)[right=of 0]{$\fin$};
        \node[action](bq)[above=1.5cm of fin]{};
        \node[state](q)[above=1.5cm of bq]{$q$};
                
        \path[->] 
        (p)   edge [above] node [left] {$a$} (ap)
        (q)   edge [above] node [right] {$a$} (aq)
              edge [above] node [right] {$b$} (bq)
        (ap)  edge [above] node [left] {$1$} (fail)
        (aq)  edge [above] node [right] {$1-x$} (p)
              edge [above] node [left] {$x$} (fin)
        (bq)  edge [above] node [left] {$x^2$} (fail)
              edge [above] node [right] {$y$} (fin)
        ;
    \end{tikzpicture}
}
\subfloat[A w\~pMDP] {
\begin{tikzpicture}[shorten >=1pt,auto,node distance=.4 cm, scale = 0.7, transform shape]
        \tikzstyle{action} = [fill=black, shape=rectangle, draw]

        \node[state,accepting](1){$1$};
        \node[action](ap)[above=1.5cm of 1]{};
        \node[state](p)[above=1.5cm of ap]{$p$};
        \node[state,accepting](0)[right=of 1]{$0$};
        \node[action](aq)[above=1.5cm of 0]{};
        \node[state,accepting](4)[right=of 0]{$4$};
        \node[action](bq)[above=1.5cm of 4]{};
        \node[state](q)[above=1.5cm of bq]{$q$};
                
        \path[->] 
        (p)   edge [above] node [left] {$a$} (ap)
        (q)   edge [above] node [right] {$a$} (aq)
              edge [above] node [right] {$b$} (bq)
        (ap)  edge [above] node [left] {} (1)
              edge [above] node [align=center] {} (q)
        (aq)  edge [above] node [left] {} (1)
              edge [above] node [left] {} (0)
        (bq)  edge [above] node [left] {} (0)
              edge [above] node [right] {} (4)
        ;
    \end{tikzpicture}
}
\subfloat[A \~pMDP] {
\begin{tikzpicture}[shorten >=1pt,auto,node distance=.4 cm, scale = 0.7, transform shape]
        \tikzstyle{action} = [fill=black, shape=rectangle, draw]

        \node[state,accepting](fail){$\fail$};
        \node[action](ap)[above=1.5cm of fail]{};
        \node[state](p)[above=1.5cm of ap]{$p$};
        \node[](0)[right=of fail]{};
        \node[action](aq)[above=1.8cm of 0]{};
        \node[state,accepting](fin)[right=of 0]{$\fin$};
        \node[action](bq)[above=1.5cm of fin]{};
        \node[state](q)[above=1.5cm of bq]{$q$};
                
        \path[->] 
        (p)   edge [above] node [left] {$a$} (ap)
        (q)   edge [above] node [right] {$a$} (aq)
              edge [above] node [right] {$b$} (bq)
        (ap)  edge [above] node [left] {$1$} (fail)
        (aq)  edge [above] node [right] {} (p)
              edge [above] node [left] {} (fin)
        (bq)  edge [above] node [left] {} (fail)
              edge [above] node [right] {} (fin)
        ;
    \end{tikzpicture}
}
\caption{Examples of all wpMDP subclasses we consider. Solid squares represent state-action pairs (see \cref{sec:graph-of-wpmdp}).} \label{fig:wpMDPs}
\end{figure}

For the rest of this section, let us fix a wpMDP $\mathcal{M} = (S,A,X,\delta,T,\rho)$.

\begin{definition}[Graph preserving valuation, from \cite{junges2021}]
  A valuation $\val: X \rightarrow \mathbb{R}$ 
  is graph preserving if the following hold for all
  $s\in (S\setminus T),s' \in S, a \in A$:
  \begin{itemize}
      \item probabilities are non-negative: $\delta(s,a,s')[\val] \geq
        0$,
      \item outgoing probabilities induce a distribution: $\sum_{s'' \in S}
        \delta(s,a,s'')[\val] = 1$,
      \item $\delta(s,a,s') \not \equiv 0 \implies \delta(s,a,s')[\val] \neq 0$.
  \end{itemize}
  The set of all such graph-preserving valuations is written as $\GraphPresVals$.
\end{definition}
A graph-preserving valuation of any wpMDP $\mathcal{M}$ gives an MDP (with real
probabilities!) that we call $\MVal$ by substituting each parameter with the corresponding real number and computing the value of the polynomials on each transition. Such an MDP is equivalent to the ones that appear in the literature \cite{baier2008,NWR_paper,junges2021}. 



\begin{definition}[Strategies]
A (memoryless deterministic) strategy $\sigma$ is a function $\sigma : (S\setminus T) \rightarrow A$ that maps states to actions. The set of all such deterministic memoryless strategies is written as $\Sigma^\mathcal{M}$.
\end{definition}


The wpMDP $\mathcal{M}$ with a valuation $\val$
and strategy $\sigma$ induce a Markov chain $\mathcal{M}^\sigma_\val=(S, \mu, T,
\rho)$ where $\mu(s,s') = \delta(s,\sigma(s),s')[\val]$ for all $s\in (S\setminus T),s' \in S$.

\remark
Note that we only consider memoryless deterministic strategies. This is because we are interested in strategies that maximise the
expected reward value for a given valuation. Since all targets
are trapping, one can prove, e.g., by reduction to
quantitative reachability or expected mean payoff~\cite{baier2008}, that
for all valuations, there is an optimal memoryless and deterministic strategy.

\subsection{The graph of a wpMDP}\label{sec:graph-of-wpmdp}
It is convenient to work with a graph representation of a
wpMDP. In that regard, we consider states $s$ (depicted as circles $\bigcirc$)
and state-action pairs $(s,a)$ (depicted as squares $\square$) to be vertices.
We call the vertices in $S_N = (S\setminus T)\times A$ \emph{nature vertices} and we will
use these to denote state-action pairs throughout the paper. We use $V = S\cup
S_N$ to denote the set of vertices of a wpMDP. The graph
$\mathcal{G}(\mathcal{M}) = (V,E)$ of a wpMDP is thus a bipartite graph with $S$ and
$S_N$ being the two partitions. We define the set of edges as follows.
\begin{align*}
    E = {} & \set{(s,(s,a))\in S\times S_N \mid \exists s'\in S\setminus \{\fail\}:\delta(s,a,s')\not\equiv0 }\\
    &{}\cup\set{((s,a),s')\in S_N\times S \mid \delta(s,a,s')\not\equiv0}
\end{align*}
That is, there is an edge from a state to a nature vertex if the nature vertex has that state as the first component of the tuple and if it can reach some state other than $\fail$ with non-zero probability; there is one from a nature vertex to a state if the polynomial on the corresponding transition is not syntactically zero.


Let
$\val \in \GraphPresVals$ be a graph-preserving valuation
and $\sigma \in \Sigma^{\mathcal{M}}$ a strategy. 
We now define the \textit{(reward) value} of a vertex in
$\mathcal{G}(\mathcal{M})$ with respect to $\val$ and $\sigma$.
\[
  \RewardValM(v) = \begin{cases}
    \sum_{s' \in vE} \delta(s,a,s')[\val] \cdot \RewardValM(s') &
    \text{if } v = (s,a) \in S_N\\
    %
    \Reward{v}{\MDPValSigma} &
    \text{otherwise}
  \end{cases}
\]
Further, we write $\RewardOptValM(v)$ for the value $\max_{\sigma \in
\Sigma^{\mathcal{M}}} \RewardValM(v)$, i.e. when the strategy is chosen to
maximise the value of the vertex. We use $\RewardVal$ and $\RewardOptVal$ when the wpMDP being referred to is clear from context.

\begin{remark}
    In the case of pMDPs and \~pMDPs, $\RewardValM(v)=\ReachProb{v}{\MDPValSigma}{\fin}$. The reader can easily verify this by looking at the definition of reward.
\end{remark}

\begin{example}
In \cref{fig:wpMDPs-a}, the valuation $\val=\set{x=0.6,y=0.28}$ 
is a graph-preserving one. We then get $\RewardOptVal(q,a)=1-0.28=0.72$ and
$\RewardOptVal(q,b)=0.6(4) = 2.4$. Hence, the optimal strategy for state $q$ will
be to choose action $b$ and $\RewardOptVal(q)=2.4$. Now, we get that
$\RewardOptVal(p)=\RewardOptVal(p,a)=2(0.28)+(2(0.6)^2-0.28)\RewardOptVal(q) = 1.616$.
\end{example}

\subsection{The never-worse relation}


Let $\mathcal{M} = (S, A, X, \delta, T, \rho)$ be a wpMDP and
$\mathcal{G}(\mathcal{M}) = (V,E)$. We will
now take the never-worse relation (NWR) \cite{NWR_paper} and generalise it to take
into account polynomials on transitions and weighted target
states.

\begin{definition}[Never-worse relation for wpMDPs]
\label{def:W_pNWR}
A subset $W \subseteq V$ of vertices is \emph{never worse} than a vertex $v
\in V$, written $v \trianglelefteq W$, if and only if:
\[
\forall \, \val \in \GraphPresVals, \exists \, w \in W : \RewardOptVal(v) \leq
\RewardOptVal(w).
\]
We write $v \sim w$ if $v \trianglelefteq \Set{w}$ and $w \trianglelefteq
\Set{v}$, and $\Tilde{v}$ for the equivalence class of $v$. 
\end{definition}
For brevity, we write $v \trianglelefteq w$ instead of $v \trianglelefteq
\Set{w}$. 
For two subsets $U,W\subseteq V$, we write $U\trianglelefteq W$ if and only if
$U \neq \emptyset$ and $u\trianglelefteq W$ for all $u\in U$.


\begin{example}
In \cref{fig:wpMDPs-a}, we have $p\trianglelefteq q$. This is because the constraints $2x^2 - y+2y=1$ and $0<2y<1$ on any graph preserving valuation ensure that $\frac{1}{2}<x$. Thus, $\RewardOptVal(q)=\RewardOptVal(q,a)>2$. Now, the value of $p$ will be a convex combination of $1$ and $\RewardOptVal(q)$. Hence, it's value will be at most the value of $q$.
\end{example}

\section{From Weighted to Non-Weighted MDPs}
\label{chap:Weighted_to_nonWeighted}
In this section, we show that we can efficiently transform any wpMDP
into a non-weighted one with a superset of states and such that the NWR is preserved.
Technically, we have two reductions: one to transform wpMDPs into pMDPs and one to transform w\~pMDPs into \~pMDPs. The former is easier but makes use of the non-trivial polynomials labelling transitions while the latter is more involved.
It follows that algorithms and complexity upper-bounds for the NWR as studied in~\cite{NWR_paper} can be applied to the trivially parametric weighted case.

\begin{remark}
\label{rem:ordering_target_weights}
Below, we will make the assumption that the target states all have distinct weights. That is $T = \{t_0, \dots, t_n\}$ with $0=\rho(t_0)<\rho(t_1)< \dots <\rho(t_n)$. This is no loss of generality. Indeed, if $t_i$ and $t_j$ have $\rho(t_i) = \rho(t_j)$ then, because all target states are trapping, we can add (on any action $a \in A$) a transition such that $\delta(t_i,a,t_j) = 1$ and remove remove $t_i$ from $T$. This can be realised using logarithmic space, does not add new states, and it clearly preserves the NWR.
\end{remark}

\subsection{Removing weights from parametric MDPs}

Let $\mathcal{M}$ be a wpMDP. We construct a pMDP $\mathcal{N}$ equivalent to $\mathcal{M}$ as described in \cref{fig:parametric-a-and-aprime}. The idea is to add transitions from the target states to freshly added $\fail$ and $\fin$ states to preserve the ratio between the values of the original target states. That is, for all states $p,q \in T \setminus \{t_0\}$, our construction guarantees the following:
\(
    {\Reward{p}{\MDPValSigma}}/{\Reward{q}{\MDPValSigma}} = {\Reward{p}{\mathcal{N}^\sigma_\val}}/{\Reward{q}{\mathcal{N}^\sigma_\val}},
\)
for all parameter valuations $\val \in \GraphPresVals$ and all strategies $\sigma \in \Sigma^{\mathcal{M}}$.
It is then easy to see, by the definition of values, that all other pairs of vertices in $\mathcal{G}(\mathcal{M})$ --- excluding $t_0$, which preserves its exact value --- also have the ratio between their values preserved.
\begin{figure}[t]
    \centering
        
    \begin{tikzpicture}[shorten >=1pt,auto,node distance=1.9 cm, scale = 0.7, transform shape]
    
        \node[]     (inv1){};
        \node[state](0)    []{$t_0$};
        \node[state](1)    [right=1.9cm of 0]{$t_1$};
        \node[state](2)    [right= of 1]{$t_2$};
        \node[]     (dots) [right=3cm of 2]{$\dots$};
        \node[state](n_1)  [right=3cm of dots]{$t_{n-1}$};
        \node[state](n)    [right= of n_1]{$t_n$};
        \node       (X)    [draw=red, fit= (inv1) (0) (n), inner sep=0.1cm, fill=red!20, fill opacity=0.2] {};
        \node[text opacity=0.2] at (X.center) {\Huge{$\mathcal{M}$}};
        
        \node[nature](x0)[below left=1cm and 1.7cm of 1]{$a_0$};
        \node[nature](x1)[right= of x0]{$a_1$};
        \node[nature](x2)[right= of x1]{$a_2$};
        \node[]      (xdots)[right=3cm of x2]{$\dots$};
        \node[nature](xn_1)[right=3cm of xdots]{$a_{n-1}$};
        \node[nature](xn)[right= of xn_1]{$a_n$};
        \node[]      (Ap)[below= of xn, text opacity=0.2]{\Huge{$\mathcal{N}$}};
        
        \node[state](fail)[below right=3cm of x1]{$\fail$};
        \node[state](fin)[right=3cm of fail]{$\fin$};
        
        \node (Y) [draw=blue, fit= (inv1) (0) (n) (fail), inner sep=0.2cm, fill=blue!20, fill opacity=0.2, dashed] {};
        
        \path[->,color=blue,dashed] 
        (0)   edge [above] node [align=center] {} (x0)
        (1)   edge [above] node [align=center] {} (x1)
        (2)   edge [above] node [align=center] {} (x2)
        (n_1) edge [above] node [align=center] {} (xn_1)
        (n)   edge [above] node [align=center] {} (xn)
        
        (x1)   edge [above, color=blue, dashed] node [above,pos=0.2] {$\rho(t_1)z$} (fin)
        (x2)   edge [above, color=blue, dashed] node [right,pos=0.2] {$\rho(t_2)z$} (fin)
        (xn_1) edge [above, color=blue, dashed] node [right,pos=0.8] {$\rho(t_{n-1})z$} (fin)
        (xn)   edge [below, color=blue, dashed, bend left] node [align=center] {$1$} (fin)
        
        (x1)   edge [below, color=blue, dashed,bend right] node [left,pos=0.2] {$1-\rho(t_1)z$} (fail)
        (x2)   edge [above, color=blue, dashed] node [left,pos=0.8] {$1-\rho(t_2)z$} (fail)
        (xn_1) edge [above, color=blue, dashed] node [left,pos=0.2] {$1-\rho(t_{n-1})z$} (fail)
        (x0)   edge [below, color=blue, dashed,bend right] node [align=center] {1} (fail)
        ;
    \end{tikzpicture}
    \caption{The pMDP $\mathcal{N}$ constructed from $\mathcal{M}$, where $z = 1/\rho(t_n)$ and $a_i$ represents the state-action pair $(t_i,a)$. The ratio of the reward values in $\mathcal{N}$ of all the target states in $\mathcal{M}$ is the same as in $\mathcal{N}$.}
    \label{fig:parametric-a-and-aprime}
\end{figure}
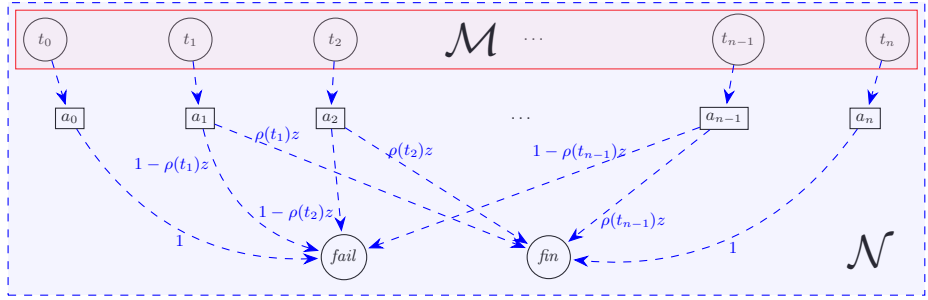

\begin{theorem}
\label{thm:w_to_b_param}
Given a wpMDP $\mathcal{M}$, let $\mathcal{G}(\mathcal{M}) = (V,E)$ and $\mathcal{N}$ be the pMDP obtained by the above construction. If $v\in V$ and $W\subseteq V$, we have $v\trianglelefteq W$ in $(\mathcal{M})$ if and only if $v\trianglelefteq W$ in $(\mathcal{N})$.
\end{theorem}

In the next section, we give an alternative reduction which preserves the property of being trivially parametric. That is, it does not add transitions with non-trivial probabilities.

\subsection{Removing weights from trivially parametric MDPs}

From a given w\~pMDP $\mathcal{M}$, we construct a \~pMDP $\mathcal{N}$ by adding to it $n+3$ vertices and $3n+2$ edges as depicted in \cref{fig:a-and-aprime}.
To simplify things, we make the assumption that any nature vertex $v\in S_N$ which does not have a path to $\fin$ has an edge to $t_0$. This is no loss of generality as such vertices can be detected in polynomial time (see, e.g.~\cite[Algorithm 46]{baier2008}) and this preserves their value.

\begin{figure}[t]
    \centering
    \begin{tikzpicture}[shorten >=1pt,auto,node distance=1.9 cm, scale = 0.7, transform shape]
    
    \node[](inv1){};
    \node[state](0)[]{$t_0$};
    \node[state](1)[right=1.9cm of 0]{$t_1$};
    \node[state](2)[right= of 1]{$t_2$};
    \node[](dots)[right=3cm of 2]{$\dots$};
    \node[state](n_1)[right=3cm of dots]{$t_{n-1}$};
    \node[state](n)[right= of n_1]{$t_n$};
    \node (X) [draw=red, fit= (inv1) (0) (n), inner sep=0.1cm, fill=red!20, fill opacity=0.2]{};
    \node[text opacity=0.2] at (X.center) {\Huge{$\mathcal{M}$}};
    
    \node[nature](x0)[below=1cm of 0]{$a_0$};
    \node[nature](x1)[below=1cm of 1]{$a_1$};
    \node[nature](x2)[below=1cm of 2]{$a_2$};
    \node[](xdots)[below=1cm of dots]{$\dots$};
    \node[nature](xn_1)[below=1cm of n_1]{$a_{n-1}$};
    \node[nature](xn)[below=1cm of n]{$a_n$};
    \node[](Ap)[below=0.8cm of xn, text opacity=0.2]{\Huge{$\mathcal{N}$}};

    \node[accepting,state](fin)[below=1cm of xdots]{$\fin$};
    \node[accepting,state](fail)[left=3cm of fin]{$\fail$};

    \node (Y) [draw=blue, fit= (inv1) (0) (n) (fail), inner sep=0.2cm, fill=blue!20, fill opacity=0.2, dashed] {};
    
    \path[->,color=blue,dashed] 
    (0)   edge [above] node [align=center] {} (x0)
    (1)   edge [above] node [align=center] {} (x1)
    (2)   edge [above] node [align=center] {} (x2)
    (n_1) edge [above] node [align=center] {} (xn_1)
    (n)   edge [above] node [align=center] {} (xn)
    
    (x1)   edge [above, color=blue, dashed] node [align=center] {} (fin)
    (x2)   edge [above, color=blue, dashed] node [align=center] {} (fin)
    (xn_1) edge [above, color=blue, dashed] node [align=center] {} (fin)
    (xn)   edge [above, color=blue, dashed] node [align=center] {} (fin)
    
    (x1) edge [above, color=blue, dashed] node [align=center] {} (0)
    (x2) edge [above, color=blue, dashed] node [align=center] {} (1)
    (xn) edge [above, color=blue, dashed] node [align=center] {} (n_1)
    
    (x0) edge [above, color=blue, dashed] node [align=center] {} (fail)
    ;
    \end{tikzpicture}
    \caption{The \~pMDP $\mathcal{N}$ constructed from $\mathcal{M}$ where $a_i$ represents the state-action pair $(t_i,a)$. The ordering of the reward values of the vertices is preserved but not the ratio.}
    \label{fig:a-and-aprime}
\end{figure}
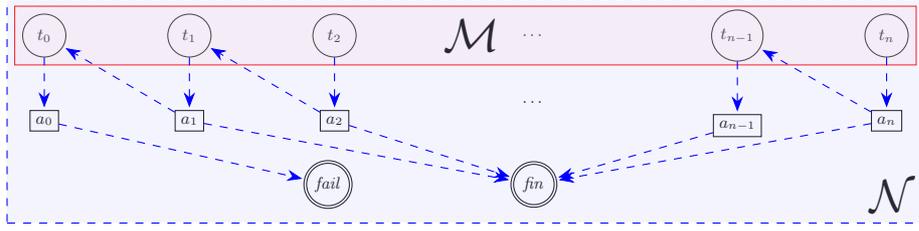

Unlike in the parametric case, the ratio of the reward values of the target vertices is not preserved but their ordering is: $\rho(t_i)<\rho(t_j)$ in $\mathcal{G}(\mathcal{M})$ if and only if
$\ReachProb{t_i}{\mathcal{N}_\mathsf{\val}^\sigma}{\fin}< \ReachProb{t_j}{\mathcal{N}_\mathsf{\val}^\sigma}{\fin}$ for all valuations $\val \in \GraphPresValsN$ and all strategies $\sigma \in \Sigma^{\mathcal{N}}$.
%
%

\begin{lemma}
\label{lem:strict-order-on-targets}
In the non-weighted \~pMDP $\mathcal{N}$, for all valuations $\val \in \GraphPresValsN$ 
and all strategies $\sigma \in \Sigma^{\mathcal{N}}$, we have $0=\ReachProb{t_0}{\mathcal{N}_\mathsf{\val}^\sigma}{\fin}<\ReachProb{t_1}{\mathcal{N}_\mathsf{\val}^\sigma}{\fin}<\dots<\ReachProb{t_n}{\mathcal{N}_\mathsf{\val}^\sigma}{\fin}$.
\end{lemma}
We show that this property, although it may not seem intuitive, is sufficient to preserve all the never-worse relations.

\begin{theorem}
\label{thm:w_to_b_trivparam}
Given a w\~pMDP $\mathcal{M}$, let $\mathcal{G}(\mathcal{M}) = (V,E)$ and $\mathcal{N}$ be the \~pMDP obtained by the above construction. If $v\in V$ and $W\subseteq V$, we have $v\trianglelefteq W$ in $(\mathcal{M})$ if and only if $v\trianglelefteq W$ in $(\mathcal{N})$.
\end{theorem}
\begin{proof}[Sketch]
We first prove the contrapositive of the ``if'' direction: Suppose $v\not\trianglelefteq W$ in the w\~pMDP $\mathcal{M}$, then there is a valuation $\val \in \GraphPresVals$ for which $\RewardOptValM(v)>\RewardOptValM(w)$ for all $w\in W$. We can now construct a valuation $\val' \in \GraphPresValsN$ so that $\Reward{*}{\mathcal{N}[\val']}(v) > \Reward{*}{\mathcal{N}[\val']}(w)$ for all $w\in W$. The existence of such a $\val'$ will show that $v\not\trianglelefteq W$ in $\mathcal{N}$, thus concluding the proof in this direction. 
More concretely, $\val'$ is obtained from $\val$
by adding probabilities to the dashed edges in $\mathcal{G}(\mathcal{N})$
so that the ratio of the reward values of states in $T$ in $\mathcal{N}$ is
the same as the ratio of their weights in \(\mathcal{M}\).

Second, we prove the contrapositive of the ``only if'' direction: Suppose $v\not\trianglelefteq W$ in the non-weighted \~pMDP $\mathcal{N}$. This means that there exists a valuation $\val \in \GraphPresVals$
such that $\RewardOptValN(v)>\RewardOptValN(w)$ for all $w\in W$. We can then use the reward values of vertices in 
$\mathcal{G}(\mathcal{N})$ (for this fixed $\val$)
to partition the vertices so that 
all vertices in a same partition have the same reward value.
This partition reveals some nice properties about the structure of $\mathcal{G}(\mathcal{N})$, and hence $\mathcal{G}(\mathcal{M})$. We use this to construct a new valuation $\val' \in \GraphPresVals$ for the original w\~pMDP $\mathcal{M}$ and show that $\Reward{*}{\mathcal{M}[\val']}(v)>\Reward{*}{\mathcal{M}[\val']}(w)$ for all $w\in W$.\qed
\end{proof}

\section{The Complexity of Deciding the NWR}
It has been shown in \cite{NWR_paper} that deciding the NWR for trivially
parametric MDPs\footnote{Technically, Le Roux and P\'erez show that the
problem is hard for \emph{target arenas}. In appendix, we give
reductions between target arenas and \~pMDPs.} is \textCoNP-complete and remains \textCoNP-hard even if the set of
actions is a singleton (i.e. the MDP is essentially a Markov chain) and even
when comparing singletons only (i.e. $v \trianglelefteq w$).
We have shown, in the previous section, that any w\~pMDP can be reduced (in
polynomial time) to a \~pMDP while preserving all the never-worse relations.
Observe that the reduction does not require adding new actions.
These imply \textCoNP-completeness for deciding the NWR in w\~pMDPs.
\begin{theorem}
    Given a w\~pMDP $\mathcal{M}$ with $\mathcal{G}(\mathcal{M}) = (V,E)$ and
    $v \in V, W \subseteq V$, determining whether $v \trianglelefteq W$ is
    \coNP-complete. Moreover, the problem is \coNP-hard even if both $W$ and the set $A$ of
    actions from $\mathcal{M}$ are singletons.
\end{theorem}

Regarding general wpMDPs, known results~\cite{monotonic,junges2021} imply that deciding the NWR is \textCoETR-hard. The \emph{existential theory of the reals} (ETR) is the set of all true sentences of the form:
\begin{equation}\label{eqn:etr}
    \exists x_1 \dots \exists x_n \varphi(x_1,\dots,x_n)
\end{equation}
where $\varphi$ is a quantifier-free (first-order) formula with inequalities
and equalities as predicates and real polynomials as terms.  The complexity
class \textETR{}~\cite{DBLP:conf/gd/Schaefer09} is the set of problems that
reduce in polynomial time to determining the truth value of a sentence like in
\eqref{eqn:etr} and \textCoETR{} is the set of problems whose complement is in
\textETR. It is known that $\NP \subseteq \ETR \subseteq \PSPACE$.

For completeness, we provide a self-contained proof of the NWR problem being
\textCoETR-hard even for pMDPs. To do so, we reduce from the
bounded-conjunction-of-inequalities problem, \textsc{BCon4Ineq} for short. It
asks, given polynomials $f_1,\dots,f_m$ of degree 4, whether there is some
valuation $\val:X\rightarrow (0,1)$ such that $\bigwedge_{i=0}^mf_i[\val]<0$?
This problem is \textETR-hard \cite[Lemma 5]{junges2021}.

\begin{theorem} \label{thm:pNWR_coETR_complete}
  Given a pMDP $\mathcal{M}$ with $\mathcal{G}(\mathcal{M}) = (V,E)$ and $v \in
  V,W\subseteq V$, determining whether $v \trianglelefteq W$ is
  \textCoETR-complete. Moreover, the problem is \textCoETR-hard even if both $W$ and the set $A$
  of actions from $\mathcal{M}$ are singletons.
\end{theorem}
\begin{proof}[Sketch]
To show \textCoETR-hardness, we first reduce \textsc{BCon4Ineq} to the problem
of deciding whether there exists a graph-preserving valuation for a given
pMDP $\mathcal{M}$. Let
$f_1,\cdots,f_m$ be the polynomials. Let $k>k_i$ where $k_i$ is the sum of the
absolute values of the coefficients of the polynomial $f_i$. We define new
polynomials $f_1',\dots,f_m'$ where $f_i'=\frac{f_i}{k}$. Note that for any
$\val:X\rightarrow (0,1)$, we have $\bigwedge_{i=0}^mf_i[\val]<0$ if and only
if $\bigwedge_{i=0}^mf'_i[\val]<0$ and we have $-1<f'_i[\val]<1$ for any such
$\val$.
\begin{figure}[t]
    \centering
        \begin{tikzpicture}[shorten >=1pt,auto,node distance=3 cm, scale = 0.7, transform shape]
            \tikzstyle{protagonist} = [shape=circle, draw]
            \tikzstyle{nature} = [shape=rectangle, draw]
                
            \node[nature](n1){$n_1$};
            \node[nature](n2)[right=of n1]{$n_2$};
            \node[nature](n3)[right=of n2]{$n_3$};
            \node[](dots)[right=of n3] {$\cdots$};
            \node[nature](nm)[right=of dots]{$n_m$};
            \node[nature](x1)[below=5cm of n1]{$v_1$};
            \node[nature](x2)[below=5cm of n2]{$v_2$};
            \node[nature](x3)[below=5cm of n3]{$v_3$};
            \node[](duts)[right=of x3]{$\cdots$};
            \node[nature](xn)[below=5cm of nm]{$v_n$};
            \node[state, accepting](fail)[below=2cm of n2]{$\fail$};
            \node[state, accepting](fin)[below=2cm of dots]{$\fin$};
            
            \path[->] (n1) edge [left,bend right] node [align=center] {$-f_1'$} (fail)
            (n2) edge [left,bend right] node [align=center] {$-f_2'$} (fail)
            (n3) edge [above] node [left,pos=0.7] {$-f_3'$} (fail)
            (nm) edge [above] node [above,pos=0.2] {$-f_m'$} (fail)
            (n1) edge [above] node [above,pos=0.2] {$1+f_1'$} (fin)
            (n2) edge [above] node [above,pos=0.2] {$1+f_2'$} (fin)
            (n3) edge [above] node [right,pos=0.2] {$1+f_3'$} (fin)
            (nm) edge [right,bend left] node [align=center] {$1+f_m'$} (fin)
            
            (x1) edge [left,bend left] node [align=center] {$x_1$} (fail)
            (x2) edge [left,bend left] node [align=center] {$x_2$} (fail)
            (x3) edge [above] node [left,pos=0.7] {$x_3$} (fail)
            (xn) edge [above] node [above,pos=0.2] {$x_n$} (fail)
            (x1) edge [above] node [below,pos=0.2] {$1-x_1$} (fin)
            (x2) edge [above] node [below,pos=0.2] {$1-x_2$} (fin)
            (x3) edge [above] node [right,pos=0.2] {$1-x_3$} (fin)
            (xn) edge [right,bend right] node [align=center] {$1-x_n$} (fin)
            ;
        \end{tikzpicture}
    \caption{The pMDP used to show \textCoETR-hardness of determining whether
    there exists a graph-preserving valuation for a pMDP.}
    \label{fig:coETR_pTA}
\end{figure}
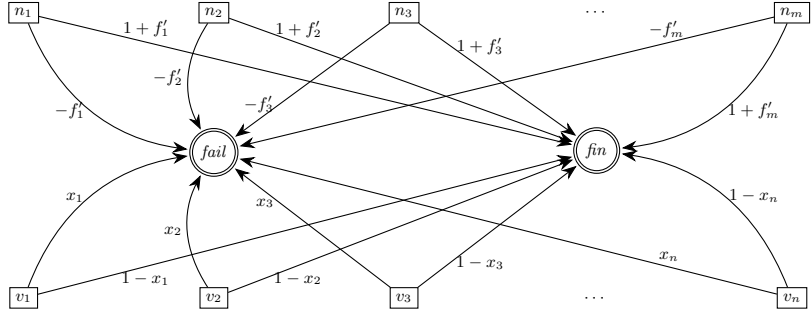
We now define a pMDP including the states and distributions shown in
\cref{fig:coETR_pTA}. All the edges from the set $\Set{ n_1, \dots, n_m }$ to
$\Set{ \fail, \fin }$ ensure that $0 < -f_i < 1$ for all $1 \leq i \leq m$
and the edges from the set $\Set{ v_1, \dots, v_n }$ to $\Set{ \fail, \fin }$
ensure that the variables can only take values from the open set $(0,1)$.
Thus, this pMDP has a graph-preserving valuation if and only if there is some
$\val : X \rightarrow (0,1)$ for which $-1 < f_i' < 0$. To conclude, we
note that $\fin \not\trianglelefteq \fail$ if and only if there is a
graph-preserving valuation for the pMDP we constructed. Observe that to obtain
the full pMDP we can add states $\Set{ s_1, \dots, s_m }$ and edges from $s_i$
to $n_i$, for all $i$. Hence, the set of actions $A$ can be a singleton.

To show that deciding the NWR is in \textCoETR, we simply encode the negation
of the NWR into a formula as in \eqref{eqn:etr} that is in ETR if and only if
the negation of the NWR holds. The encoding is quite natural and can be seen
as a ``symbolic'' version of the classical linear programs used to encode MDP
values.\qed
\end{proof}

We conclude this section with a discussion on the NWR equivalence relation.
\subsection{The complexity of deciding NWR equivalences}
Unfortunately, deciding NWR equivalences is just as hard as deciding the NWR
in general. To prove this, we give a small gadget which ensures that $v \sim
w$ if and only if some NWR holds.
\begin{theorem}\label{thm:equivalence_complexity}
  Given a pMDP $\mathcal{M}$ with $\mathcal{G}(\mathcal{M}) = (V,E)$ and $v,w
  \in V$, determining whether $v \sim w$ is \textCoETR-hard. If $\mathcal{M}$
  is a \~pMDP then the problem is \textCoNP-hard in general and in \P{} if its
  set of actions $A$ is a singleton.
\end{theorem}
To show that NWR equivalences are decidable in polynomial time for \~pMDPs
with one action, we argue that $u \sim w$ holds if and only if there exists a
unique $z \in V$ such that
$\ReachProb{u}{\mathcal{C}}{z}=\ReachProb{w}{\mathcal{C}}{z}=1$. This
characterisation allows us to compute all equivalence classes based on where
every state is almost-surely reachable from. The latter can be done in
polynomial time~\cite[Algorithm 45]{baier2008}.

\section{Action Pruning via the NWR}

In this section, we show various ways to efficiently under-approximate the NWR
--- that is, obtain subsets of the relation --- and use these
under-approximations to reduce the size of the wpMDP by pruning actions and by
collapsing equivalent states of a wpMDP. As we have shown in
\cref{chap:Weighted_to_nonWeighted}, any weighted w\~pMDP can be efficiently
converted into a (non-weighted) \~pMDP while preserving all the never-worse
relations from the original weighted model. Hence, we present
under-approximations for non-weighted models only (i.e.,
$T=\Set{\fin,\fail}$). We will also focus on trivially parametric MDPs only.
Note that if $U\trianglelefteq W$ in a \~pMDP, then $U\trianglelefteq W$ in
any pMDP having the same underlying graph as the first \~pMDP. (That is, as
long as it admits a graph-preserving valuation.)

Henceforth, let $\mathcal{M} = (S,A,\Delta,T)$ be a non-weighted trivially
parametric MDP and write $\mathcal{G}(\mathcal{M}) = (V,E)$. We make the assumption that all \emph{extremal-value vertices} \cite{NWR_paper} have been
contracted: An extremal-value vertex $v\in V$ is such that either
$\RewardOptVal(v)=0$ or $\RewardOptVal(v)=1$ holds for all $\val \in
\GraphPresVals$. The set of all such vertices can be computed in polynomial
time \cite[Algorithm 45]{baier2008}. One can then contract all value-$1$
vertices with $\fin$ and all value-$0$ with $\fail$, without changing the
values of the other vertices.


\subsection{The under-approximation graph}
We will represent our current approximation of the NWR by means of a
(directed) \emph{under-approximation graph} $\mathcal{U} = (N,R)$ such that $N
\subseteq 2^V$. For $U,W \in N$, let $U \xrightarrow{\mathcal{U}} W$ denote the fact
that there is a path in $\mathcal{U}$ from $U$ to $W$. Throughout our
algorithm we will observe the following invariant.
\begin{equation}\label{eqn:underapprox}
  \forall U,W \in N : U \xrightarrow{\mathcal{U}} W \implies U \trianglelefteq W
\end{equation}

\paragraph{Initializing the graph.} For our initial under-approximation of the
NWR, we construct $\mathcal{U}$ with vertex set $N = \{ \{v\}, vE \mid v \in V\}$. That is, it contains all states, all
state-action pairs, and all sets of immediate successors of states or state-action pairs. For the edges, we
add them so that the following hold.
\begin{enumerate}
    \item $R \ni (\{\fail\}, n)$ for all $n \in N$
    \item $R \ni (n, \{\fin\})$ for all $n \in N$
    \item $R \ni (\{v\}, vE)$ for all $v \in V$
    \item $R \ni (vE, \{v\})$ for all $v \in S$
\end{enumerate}

We have shown that this initialization yields a correct under-approximation.
\begin{lemma}\label{lem:underapprox_init}
    Let $\mathcal{U}$ be the initial under-approximation graph as defined above.
    Then, invariant \eqref{eqn:underapprox} holds true.
\end{lemma}
Below, we describe how the under-approximation graph can be updated to get ever tighter under-approximations of the NWR.

\paragraph{Updating and querying the graph.}
Whenever we add a vertex $U$ to some under-approximation graph $\mathcal{U}$, we also add:
\begin{itemize}
    \item an edge from $\{\fail\}$ and each $U' \in N$ such that $U' \subset U$ to $U$,
    \item an edge from $U$ to $\{\fin\}$ and each $U'' \in N$ such that $U \subset U''$, and
    \item an edge from $U$ to $W$ if $\{u\} \xrightarrow{\mathcal{U}} W$ holds for all $u \in U$.
\end{itemize} 
When we add an edge $(U,W)$ to $\mathcal{U}$, we first add the vertices $U$
and $W$ as previously explained, if they are not yet in the graph. Finally, to
query the graph and determine whether $U \xrightarrow{\mathcal{U}} W$, we
simply search the graph for $W$ from $U$ as suggested by the definitions.

\begin{lemma}\label{lem:underapprox_add_edge}
    Let $\mathcal{U}'$ be an under-approximation graph satisfying invariant 
    \eqref{eqn:underapprox} and $\mathcal{U}$ the under-approximation graph
    obtained after adding an edge $(U,W) \in 2^V \times 2^V$ as
    described above. If $U \trianglelefteq W$ then $\mathcal{U}$ satisfies invariant
    \eqref{eqn:underapprox}.
\end{lemma}
Before we delve into how to infer pairs from the NWR to improve our
approximation graph, we need to recall the notion of \emph{end component}.

\subsection{End components and quotienting}
Say $(P,B)$ is a sub-MDP of $\mathcal{M}$, for $P\subseteq S$ and $B\subseteq \{(p,a) \mid (p,(p,a)) \in E\}$, 
if:
\begin{itemize}
  \item for all $p\in P$ there is at least one $a\in A$ such that $(p,a)\in B$, and
  \item for all $p'\in S$ and $(p,a) \in B$ with $((p,a),p')\in E$ we have $p'\in P$.
\end{itemize}
One can think of sub-MDPs as connected components of the original MDP obtained by removing some actions from some states.

A sub-MDP $(P,B)$ of $\mathcal{M}$ is an \emph{end component} if the subgraph
of $\mathcal{G}(\mathcal{M})$ induced by $P \cup B$ is strongly connected.
The end component $(P,B)$ is maximal (a MEC) if there is
no other end component $(P',B')\neq(P,B)$ such that $B \subseteq B'$. The set
of MECs can be computed in polynomial time \cite[Algorithm 47]{baier2008}.

It is known that end components are a special case of the NWR. In fact, all
vertices in a same end component are NWR-equivalent~\cite[Lemma 3]{NWR_paper}.
Hence, as we have done for extremal-value states, we can assume they have been
contracted. It will be useful, in later discussions, for us to make explicit
the \emph{quotienting} construction realizing this contraction.

\paragraph{MEC quotient.}
Let $s \in S$ and write $\mecclass{s}$ to denote $\{s\}$ if $s$ is not part of
any MEC $(P,B)$ and $P$ otherwise to denote the unique MEC
$(P,B)$ with $s \in P$. Now, we denote by $\mmec$ the \emph{MEC-quotient} of
$\mathcal{M}$. That is, $\mmec$ is the \~pMDP $\mmec = (S',A,\Delta',T')$
where:
\begin{itemize}
    \item $T'=\Set{\mecclass{t} \mid \exists t\in T}$,
    \item $S'=\Set{\mecclass{s} \mid \exists s\in S}$, and
    \item $\Delta' = \{(\mecclass{s},a,\mecclass{s'}) \mid \exists (s,a,s'')
      \in \Delta, \mecclass{s} \neq \mecclass{s''}\}$.
\end{itemize}
The construction does preserve the never-worse relations amongst all states.
\begin{proposition}
    For all $U,W \subseteq S$ we have that
    \(
        U \trianglelefteq W
    \)
    in $\mathcal{G}(\mathcal{M})$
    if and only if
    \(
        \{\mecclass{u} \mid \exists u \in U\} \trianglelefteq \{\mecclass{w} \mid \exists w \in W\}
    \) in $\mathcal{G}(\mmec)$.
\end{proposition}

We now describe how to prune sub-optimal actions while updating the approximation graph via inference rules for the NWR.


\subsection{Pruning actions and inferring the NWR}
The next theorem states that we can prune actions that lead to sub-optimal vertices with respect to the NWR once we have taken the MEC-quotient.

\begin{theorem}\label{thm:mec_edge_removal}
Let $\mathcal{M}$ be a \~pMDP such that $\mathcal{M}$ is isomorphic with $\mmec$. Then, for all valuations $\val \in \GraphPresVals$ and all $(s,a) \in S\times A$ with $(s,a) \trianglelefteq (sE \setminus \Set{(s,a)})$:
\[
    \max_\sigma\ReachProb{s}{\mathcal{M}_{\sigma}^\val}{\fin}=\max_{\sigma'}\ReachProb{s}{\mathcal{N}_{\sigma'}^{\val}}{\fin},
\]
where $\mathcal{N} = (S, A, \Delta \setminus \Set{(s, a, s') \in \Delta \mid \exists s' \in S}, T)$.
\end{theorem}

It remains for us to introduce inference rules to derive new NWR pairs based on our approximation graph $\mathcal{U}$. The following definitions will be useful.

\paragraph{Essential states and almost-sure reachability.} We say the set $W
\subseteq V$ is \emph{essential for}\footnote{This definition is inspired
by~\cite{DBLP:conf/papm/DArgenioJJL01} but it is not exactly the same as in
that paper.} $U \subseteq V$, written $U\sqsubseteq W$, if each path starting
from any vertex in $U$ and ending in $\fin$ contains a vertex from $W$. (Intuitively, removing $W$ disconnects $U$ from $\fin$.) We say
that $U \subseteq V$ \emph{almost-surely reaches} $W \subseteq V$, written $U
\xrightarrow{\mathrm{a.s.}} W$, if all states $u \in U$ become value-$1$
states after replacing $T$ with $W \cup \{\fail\}$ and making all $w \in W$
have weight $1$. Recall that value-$1$ states are so for all graph-preserving
valuations and that they can be computed in polynomial time. Similarly,
whether $W$ is essential for $U$ can be determined in polynomial time by
searching in the subgraph of $\mathcal{G}(\mathcal{M})$ obtained by removing $W$.
We also define, for $W \subseteq V$, a function $f_{W} : 2^{S} \to 2^{S}$:
\[
    f_{W}(D) = D \cup \left\{z \in V \:\middle|\: \{z\} \xrightarrow{\mathcal{U}} W \cup D\right\}
\]
and write $\mu D.f_{W}(D)$ for its least fixed point with respect to the
subset lattice.

\paragraph{Inference rules.}
For all $U,W \subseteq V$, the following rules hold true. It is worth
mentioning that these are strict generalisations of~\cite[Propositions 1, 2]{NWR_paper}.
\begin{equation}\label{eqn:prop1}    
  \infer{U \trianglelefteq W}{U \sqsubseteq \mu D. f_W(D)}
  \quad
  \infer{U \trianglelefteq W}{\exists w \in W : \{w\}
  \xrightarrow{\mathrm{a.s}} \{s \in S \mid U \xrightarrow{\mathcal{U}}
  \{s\}\}}
\end{equation}

\begin{lemma}
    Let $\mathcal{U}$ be an approximation graph satisfying
    \eqref{eqn:underapprox}. Then, the inference rules from \eqref{eqn:prop1}
    are correct.
\end{lemma}

Our proposed action-pruning method is given as Algorithm~\ref{alg:ua_2}. Note
that in line~\ref{lin:contract} we do a final contraction of NWR equivalences.
These are the equivalences that can be derived between any two singletons $\{x\}$ and $\{y\}$ that lie on a same cycle in $\mathcal{U}$. The algorithm clearly
terminates because the NWR is finite, there is no unbounded recursion, and there are no unbounded loops. Finally, correctness follows from the results in this
section and, for the aforementioned contraction of NWR equivalence classes,
from~\cite[Theorem 1]{NWR_paper}.

\begin{algorithm}[hbt]
\caption{Reduction using the under-approximation}\label{alg:ua_2}
\Input{A \~pMDP $\mathcal{M} = (S,A,\Delta,T)$}
\Output{A (hopefully) smaller \~pMDP}
    Initialize $\mathcal{U}$\;
    \Repeat{$\mathcal{U}$ is unchanged}{\label{line:inner-loop}
        \For{$s\in S$}{
            \For{$(s,a) \in sE$}{
                $W \gets sE \setminus \{(s,a)\}$\;
                \uIf{$\{(s,a)\} \xrightarrow{\mathcal{U}} W$}{
                  Prune all $\{(s,a,s') \in \Delta \mid \exists s' \in S\}$\;
                }\uElseIf{$\{(s,a) \} \sqsubseteq \mu D.f_W(D)$}{
                  Add edge $((s,a),W)$ to $\mathcal{U}$\;
                  Prune all $\{(s,a,s') \in \Delta \mid \exists s' \in S\}$\;
                }\uElseIf{$\exists w \in W : \{w\} \xrightarrow{\mathrm{a.s.}} \{s' \in S \mid \{(s,a)\} \xrightarrow{\mathcal{U}} \{s'\}\}$}{
                  Add edge $((s,a),W)$ to $\mathcal{U}$\;
                  Prune all $\{(s,a,s') \in \Delta \mid \exists s' \in S\}$\;
                }                   
            }
        }
    }
    Contract NWR equivalences from cycles in
    $\mathcal{G}$\;\label{lin:contract}
    Return the new \~pMDP\;
\end{algorithm}

\section{Experiments}

In this section, we show the results of the experiments that we have conducted to
assess the effectiveness of the techniques described in the previous sections.
We first motivate the research questions that we wish to answer using the
experiments. Then, we describe the models that we have used during the experiments,
including the pre-processing steps that were performed. Next, we describe the
different setups in which our techniques have been applied to the selected and
pre-processed models. Finally, we discuss the results that we have obtained during
the experiments. Everything needed to re-run our experiments is in~\cite{NWR_code}.






We approach our experiments with the following research questions in mind.
\begin{itemize}
    \item[Q1.] Is the NWR really useful in practice?
    \item[Q2.] Does our under-approximation give us significant reductions in the size of some models? How do these reductions compare to the known reductions?
    \item[Q3.] Does the under-approximation yield reductions in early iterations?
\end{itemize}
Note that the first question is (intentionally) vague. To approach it, we thus focus on the more concrete two that follow. Namely, for the under-approximation, we want to know whether we can reduce the size of common benchmarks and, moreover, whether such reduction in size is obtained through a small number of applications of our inference rules. On the one hand, a positive answer to Q2 would mean the NWR is indeed useful as a preprocessing step to reduce the size of an MDP. On the other hand, a positive answer to Q3 would imply that a small number of applications of our inference rules is always worth trying.

\subsection{Benchmarks and protocol}


For the experiments, we have considered all discrete-time MDPs encoded as PRISM models~\cite{DBLP:conf/cav/KwiatkowskaNP11} from the Quantitative Verification Benchmark set \cite{DBLP:conf/isola/BuddeHKKPQTZ20}, provided that they have at least one maximal reachability property. From these, we have filtered out those that have an unknown number of states, or more than $64.000$ states, as specified on the QComp website, in order to limit the runtime of the experiments.


We have added every maximal reachability property from the \texttt{props}
files as labels in the model files. In some cases, either out of curiosity, or
to compare with other known reductions, we have introduced new labels to the models
as final states. For example, in the \texttt{wlan\_dl} benchmark, instead of
checking the minimum probability of both stations sending correctly within
deadline, we check the maximum probability of doing so. In the
\texttt{consensus} benchmarks, we add new labels \texttt{fin\_and\_all\_1} and
\texttt{fin\_and\_not\_all\_1} to compare our reductions to the ones used by
Bharadwaj et al. \cite{bharadwaj2017}. Some of these labels used in the tables and graphs later are listed below:
%
\begin{enumerate}
    \item Model: \texttt{consensus}. Labels: \texttt{disagree}, \texttt{fin\_and\_all\_1}, \texttt{fin\_and\_not\_all\_1}. The label \texttt{disagree} is one of the properties in the property file and the other two labels are meant for comparisons with past results \cite{bharadwaj2017}.
    \item Model: \texttt{zeroconf\_dl}. Label: \texttt{deadline\_min}. The minimum probability of not finishing within the deadline does not look like a maximum reachability probability question, but it turns out to be the same as the maximum probability of reaching a new IP address within the deadline.
    \item Model: \texttt{zeroconf}. Label: \texttt{correct\_max}.
    \item Model: \texttt{crowds}. Label: \texttt{obs\_send\_all}. 
    \item Model: \texttt{brp}. Labels: \texttt{no\_succ\_trans}, \texttt{rep\_uncertainty}, \texttt{not\_rec\_but\_sent}.
\end{enumerate}

The PRISM files, with the added labels, have then converted into an explicit JSON format using STORM \cite{DBLP:journals/sttt/HenselJKQV22}. This format was chosen since it explicitly specifies all states and state-action pairs in an easy to manipulate file format.

We have then applied two pre-processing techniques. Namely, we have collapsed all extremal states and state-action pairs and collapsed all MECs. These two techniques, and the way they relate to the never-worse relation, are described in the previous sections. Some models were completely reduced after collapsing all extremal states, with only the $\fin$ and $\fail$ states remaining. We have not included such MDPs in the discussion here onward, but they are listed in~\cite{NWR_code}. Collapsing the MECs does not prove to be as effective, providing no additional reductions for the models under consideration.





\subsection{Different setups}

In each experiment, we repeatedly apply Algorithm~\ref{alg:ua_2} on each
model. If for two consecutive iterations, the collapse of equivalent states
(see line~\ref{lin:contract}) has no effect then we terminate the
experiment. Each iteration, we start with an empty under-approximation
graph, to avoid states removed during the collapse being included in
the under-approximation. This is described in
Algorithm~\ref{alg:experimental_setup}.

\begin{algorithm}
\caption{The experimental setup}\label{alg:experimental_setup}
\Input{A \~pMDP $\mathcal{M} = (S,A,\Delta,T)$}
\Output{A \~pMDP that can be analysed}
    Contract extremal-value states\;
    Contract MECs\;
    \Repeat{collapsing NWR equivalences has no effect}{\label{line:outer-loop}
        Start with an empty under-approximation $\mathcal{U}$\;
        Apply Algorithm~\ref{alg:ua_2}\;
    }
\end{algorithm}

In order to reduce the runtime of the experiments, we consider a maximal number of iterations for the loops specified on line \ref{line:inner-loop} of Algorithm~\ref{alg:ua_2} and line \ref{line:outer-loop} of  Algorithm~\ref{alg:experimental_setup}. We refer to the loop of Algorithm~\ref{alg:ua_2} as the \say{inner loop}, and to the loop of Algorithm~\ref{alg:experimental_setup} as the \say{outer loop}.
During the experiments, we found that running the inner loop during the first iteration of the outer loop was too time-consuming. We therefore decide to skip the inner loop during the first iteration of the outer loop. Even then, running the inner loop turns out to be impractical for some models, resulting in the following two experimental setups:
\begin{enumerate}
    \item For some models, we allow the outer loop to run for a maximum of 17 times. In our experiments, no benchmark required more than 7 iterations. The inner loop is limited to 3 iterations.
    \item For others, we ran the outer loop 3 times without running the inner loop.
\end{enumerate}

The Markov chains considered in our experiments are always used in the second experimental setup because Markov chains are MDPs with a single action for each state, making the action-pruning inner loop redundant.

\begin{remark}
    Every time we add a node $U$ to the under-approximation graph, we skip the step ``an edge from $U$ to $W$ if $\{u\} \xrightarrow{\mathcal{U}} W$ holds for all $u \in U$'' since we found that this step is not efficient.
\end{remark}

\subsection{Results, tables and graphs}

%
%
%
%
%
\begin{table}[h!]
\caption{Table with the number of states and choices during each step of running the benchmarks. The last column has the running time after preprocessing. The experiments with (gray) colored running time were run with the first experimental setup.}
\centering
\resizebox{\columnwidth}{!}{%
\begin{tabular}{c c|c c|c c|c c c}
 \hline
 & & \multicolumn{2}{c|}{original size} & \multicolumn{2}{c|}{preprocessing} & \multicolumn{3}{c}{under-approx}\\
 benchmark & instance & \#st & \#ch & \#st & \#ch & \#st & \#ch & time\\ 
 \hline\hline
 \multirow{4}{*}{\makecell{\texttt{consensus}\\ \texttt{"disagree"}\\(N,K)}}
                      & (2,2) & 274 & 400 & 232 & 344 & 148 & 260 & \CC{40}78.17s\\
                      & (2,4) & 530 & 784 & 488 & 728 & 308 & 548 & \CC{40}383.13s\\
                      & (2,8) & 1042 & 1552 & 1000 & 1496 & 628 & 1124 & 19.38s\\
                      & (2,16) & 2066 & 3088 & 2024 & 3032 & 1268 & 2276 & 78.47s\\\hline
\multirow{4}{*}{\makecell{\texttt{consensus}\\ \texttt{"fin\_and\_all\_1"}\\(N,K)}} 
                      & (2,2) & 274 & 400 & 173 & 280 & 127 & 232 & \CC{40}124.33s\\
                      & (2,4) & 530 & 784 & 365 & 600 & 271 & 504 & \CC{40}674.17s\\
                      & (2,8) & 1042 & 1552 & 749 & 1240 & 561 & 1052 & 14.14s\\
                      & (2,16) & 2066 & 3088 & 1517 & 2520 & 1137 & 2140 & 57.78s\\\hline
\multirow{4}{*}{\makecell{\texttt{consensus}\\ \texttt{"fin\_and\_not\_all\_1"}\\(N,K)}} 
                      & (2,2) & 274 & 400 & 165 & 268 & 123 & 226 & \CC{40}61.15s\\
                      & (2,4) & 530 & 784 & 357 & 588 & 267 & 498 & \CC{40}317.95s\\
                      & (2,8) & 1042 & 1552 & 741 & 1228 & 555 & 1042 & 14.00s\\
                      & (2,16) & 2066 & 3088 & 1509 & 2508 & 1131 & 2130 & 58.24s\\\hline
\multirow{6}{*}{\makecell{\texttt{zeroconf\_dl}\\\texttt{"deadline\_min"}\\(N,K,reset,deadline)}} 
                      & (1000,1,true,10) & 3837 & 4790 & 460 & 552 & 242 & 333 & \CC{40}98.47s\\
                      & (1000,1,true,20) & 7672 & 9775 & 2709 & 3351 & 1313 & 1945 & 79.49s\\
                      & (1000,1,true,30) & 11607 & 14860 & 4999 & 6211 & 2413 & 3605 & 266.00s\\
                      & (1000,1,true,40) & 15642 & 20045 & 7289 & 9071 & 3513 & 5265 & 560.82s\\
                      & (1000,1,true,50) & 19777 & 25330 & 9579 & 11931 & 4613 & 6925 & 1013.18s\\
                      & (1000,1,false,10) & 12242 & 18200 & 460 & 552 & 242 & 333 & \CC{40}95.69s\\\hline
\multirow{4}{*}{\makecell{\texttt{zeroconf}\\\texttt{"correct\_max"}\\(N,K,reset)}} 
                      & (1000,2,true) & 672 & 814 & 388 & 481 & 118 & 150 & \CC{40}29.13s\\
                      & (1000,4,true) & 1090 & 1342 & 674 & 855 & 206 & 260 & \CC{40}76.06s\\
                      & (1000,6,true) & 1508 & 1870 & 960 & 1229 & 294 & 370 & \CC{40}161.01s\\
                      & (1000,8,true) & 1926 & 2398 & 1246 & 1603 & 382 & 480 & \CC{40}237.69s\\\hline
\multirow{8}{*}{\makecell{\texttt{crowds}\\\texttt{"obs\_send\_all"}\\(TotalRuns,CrowdSize)}} 
                      & (3,5)  & 1200 & 1198 & 268 & 266 & 182 & 180 & 1.03s\\
                      & (4,5)  & 3517 & 3515 & 1030 & 1028 & 632 & 630 & 11.93s\\
                      & (5,5)  & 8655 & 8653 & 2930 & 2928 & 1682 & 1680 & 87.76s\\
                      & (6,5)  & 18819 & 18817 & 6905 & 6903 & 3782 & 3780 & 467.08s\\
                      & (3,10) & 6565 & 6563 & 833 & 831 & 662 & 660 & 13.12s\\
                      & (4,10) & 30072 & 30070 & 5430 & 5428 & 3962 & 3960 & 481.47s\\
                      & (3,15) & 19230 & 19228 & 1698 & 1696 & 1442 & 1440 & 69.90s\\\hline
\multirow{4}{*}{\makecell{\texttt{brp}\\ \texttt{"no\_succ\_trans"}\\(N,MAX)}} 
                      & (64,2) & 2695 & 2693 & 1982 & 1980 & 768 & 766 & 26.76s\\
                      & (64,3) & 3528 & 3526 & 2812 & 2810 & 1087 & 1085 & 53.80s\\
                      & (64,4) & 4361 & 4359 & 3642 & 3640 & 1406 & 1404 & 89.44s\\
                      & (64,5) & 5194 & 5192 & 4472 & 4470 & 1725 & 1723 & 134.57s\\\hline
\multirow{4}{*}{\makecell{\texttt{brp}\\ \texttt{"rep\_uncertainty"}\\(N,MAX)}} 
                      & (64,2) & 2695 & 2693 & 1982 & 1980 & 768 & 766 & 25.74s\\
                      & (64,3) & 3528 & 3526 & 2812 & 2810 & 1087 & 1085 & 52.20s\\
                      & (64,4) & 4361 & 4359 & 3642 & 3640 & 1406 & 1404 & 87.62s\\
                      & (64,5) & 5194 & 5192 & 4472 & 4470 & 1725 & 1723 & 133.96s\\\hline
\multirow{4}{*}{\makecell{\texttt{brp}\\ \texttt{"not\_rec\_but\_sent"}\\(N,MAX)}} 
                      & (64,2) & 2695 & 2693 & 8 & 6 & 5 & 3 & 0.004s\\
                      & (64,3) & 3528 & 3526 & 10 & 8 & 6 & 4 & 0.01s\\
                      & (64,4) & 4361 & 4359 & 12 & 10 & 7 & 5 & 0.01s\\
                      & (64,5) & 5194 & 5192 & 14 & 12 & 8 & 6 & 0.01s\\\hline
\end{tabular}
}
\label{table:tab_ua}
\end{table}

\begin{remark}
    We use the term \textit{choice} to denote a state-action pair in the MDP.
\end{remark}

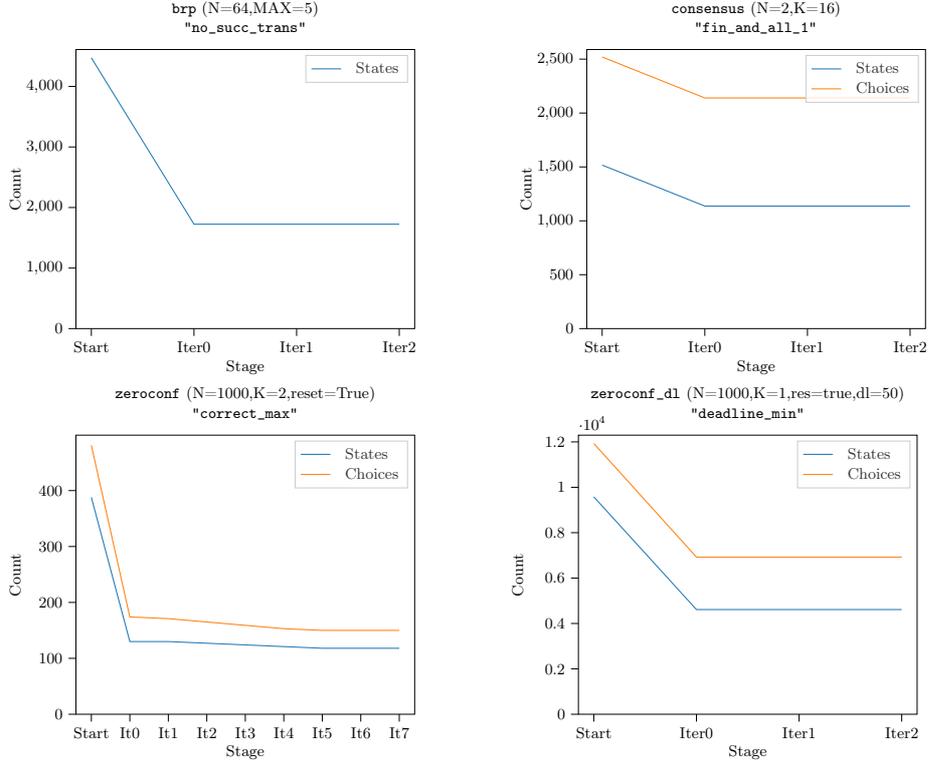
\begin{figure}[t]
\centering
\input{images_tikz/plot_tikz_43.tikz}\hfill
\input{images_tikz/plot_tikz_4.tikz}\\
\input{images_tikz/plot_tikz_88.tikz}\hfill
\input{images_tikz/plot_tikz_9.tikz}
\caption{Some plots visualizing the reductions obtained per iteration.}
\label{fig:plots}
\end{figure}


\cref{table:tab_ua} shows the state-space reduction in some of the benchmarks we ran. 
After preprocessing, we obtain a 
$61\%$ reduction on average in all the $\texttt{brp}$ benchmarks with
\texttt{"no\_succ\_trans"} and \texttt{"rep\_uncertainty"} as final states and
a $40\%$ reduction on average with \texttt{"not\_res\_but\_sent"} as final
states. In the \texttt{zeroconf\_dl} benchmarks, we obtain a further $50\%$
reduction on average and around $69\%$ reduction in all \texttt{zeroconf}
benchmarks. In \texttt{consensus} benchmarks with \texttt{"disagree"} final
states, we get around $37\%$ further reduction after the preprocessing only
removed less than $16\%$ of the states. The rest of the data (including
further benchmarks) can be found in \cite{NWR_code}. All of these
observations\footnote{We were not able to properly compare our
  results to the ones in~\cite{bharadwaj2017}. The sizes reported therein do
  not match those from the QComp models. The authors
  confirmed theirs are based on modified models of which the data and code
have been misplaced.}
point to Q2 having a positive answer. 

Among all the benchmarks we tested, only \texttt{zeroconf} and some instances
of \texttt{consensus} and \texttt{zeroconf\_dl} see any improvements after the
first iteration of the algorithm \cite{NWR_code}. Even then, the
majority of the improvement is seen in the first iteration, the improvement
in the later iterations is minimal, as depicted in \cref{fig:plots}, and no
benchmark runs for more than 7 iterations. The \texttt{consensus} benchmark
with \texttt{"disagree"} final states runs with the second setup for
$(N=2,K=8),(N=2,K=16)$ and with the first setup for $(N=2,K=2),(N=2,K=4)$. As
seen in \cref{table:tab_ua}, all the instances gave us around a $36\%$ further
reduction after preprocessing, although the ones where the algorithm ran only
for one iteration are significantly faster. This gives us a positive answer
for Q3.

\section{Conclusions}
We have extended the never-worse relation to the quantitative reachability
setting on parametric Markov decision processes. While the complexity of
deciding the relation is relatively high (\textCoETR-complete), efficient
under-approximations seem promising. We believe that the relation could be
made more applicable by exploring the computation of such under-approximations
in an on-the-fly fashion such as, for instance, on-the-fly MEC algorithms~\cite{DBLP:conf/atva/BrazdilCCFKKPU14}. The following questions also warrant futher (empirical) study: What kind of MDPs have large reductions under the NWR? What do the reduced MDPs look like? What kinds of states are detected and removed by the NWR under-approximations? 

As additional future work,
we think it is worthwhile to study approximations that are (more)
efficient when the MDP is encoded as a reduced ordered binary decision diagram
(BDD) or in a particular modelling language such as the PRISM language. For the latter, an initial work in this direction is~\cite{DBLP:conf/vmcai/WinklerLK22}.



\subsubsection{Acknowledgements.} 
This work was supported by the Belgian FWO
``SAILor'' (G030020N) and Flemish inter-university (iBOF) ``DESCARTES'' projects.

\bibliographystyle{splncs04}
\bibliography{bibliography}

\appendix
\clearpage
\section{Definitions and preliminaries}

We give an alternate equivalent definition of $\RewardOptVal(v)$ by using Bellman equations which will be useful in proofs:
\begin{lemma}\label{lem:bellman_eq}
    Consider a wpMDP $\mathcal{M} = (S,A,X,\delta,T,\rho)$. Let $\boldsymbol{Z}$ denote the set of all vertices that cannot reach $T$. The vector $(x_v)_{v\in V}$ with $x_p=\RewardOptValM(v)$ is the unique solution of the following system of equations:
    \begin{itemize}
        \item If $v\in T$ then $x_v=\rho(v)$;
        \item If $v\in \boldsymbol{Z}$ then $x_v=0$;
        \item If $v\in S_N \setminus (T\cup \boldsymbol{Z})$ then $x_v=\sum_{s\in vE}\delta(v,s)[\val]\cdot x_s$;
        \item If $v\in S \setminus (T\cup \boldsymbol{Z})$ then $x_v=\max_{a\in vE}x_a$.
    \end{itemize}
\end{lemma}

The following lemma observes the relation between the optimal rewards of states and state-action pairs to their (immediate) successors.
\begin{lemma}\label{lem:val_compare_child}
For a wpMDP $\mathcal{M} = (S,A,X,\delta,T,\rho)$ and a valuation $\val\in\GraphPresVals$, for all $u\in S$ we have
\[\RewardOptVal(u)=\max_{a\in A} \RewardOptVal(u,a),\]
and for all $u\in S_N$
\[\min_{(u,a)\in uE}\RewardOptVal(u,a)\leq \RewardOptVal(u)\leq \max_{(u,a')\in uE}\RewardOptVal(u,a'),\]
where the equality holds iff all the vertices in $uE$ have the same value.
\end{lemma}
The proof directly follows from the equivalent definition of reward values.

\section{Removing weights from parametric MDPs}
Let $\mathcal{M} = (S, A, X, \delta, T, \rho)$ be a wpMDP with target states $T=\Set{ t_0, t_1, \cdots, t_n }$ such that $0=\rho(t_0)<\rho(t_1)<\cdots<\rho(t_n)$. We will now create a (non-weighted) pMDP $\mathcal{M'}=(S',A', X',\delta', T')$, which has the same never-worse relations as the given wpMDP $\mathcal{M}$. We do so by adding a new action $a\not\in A$ which ensures that the ratio of the weights of the target states in $\mathcal{M}$ is the same as the ratio of the optimal reward values of those states in $\mathcal{N}$. Formally:
\[
\begin{array}{rl}
    S' &= S \uplus \Set{\fail,\fin}, \\
    A' &= A \uplus \Set{a}, \\
    \delta'(s,a',s')   &= \delta(s,a',s') \text{ for $s,s'\in S$ and $a'\in A$},\\
    \delta'(t,a,\fin)  &= \frac{\rho(t)}{\rho(t_n)} \text{ for $t\in T$},\\
    \delta'(t,a,\fail) &= 1-\frac{\rho(t)}{\rho(t_n)} \text{ for $t\in T$},\\
    X'   &= X,\\
    T'   &= \Set{\fail,\fin}.\\
\end{array}
\]

\subsection{Proof of \cref{thm:w_to_b_param}}

Let $c=\frac{1}{\rho(t_n)}$. We will show that $\RewardOptVal(v)_\mathcal{M}\cdot c=\RewardOptVal(v)_\mathcal{N}$ for any $v\in V$, and for any graph-preserving valuations $\mathsf{val}$ for $\mathcal{M}$ and $\mathcal{N}$. Proving this will prove the theorem. 

First, we note that $\RewardOptVal(t)=\rho(t)\cdot c$ for $t\in T$. Next, we note that any path from any state $s\in S$ to $\fin$ in $\mathcal{N}$ must go through $T$. Thus, for any state $s\in S$, we have:
\begin{align*}
\RewardOptValN(u)&=\max_{\sigma\in\Sigma^\mathcal{N}}\RewardValN(u)\\
&=\max_{\sigma\in\Sigma^\mathcal{N}}\Reward{u}{\MDPValSigmaN}\\
&=\max_{\sigma\in\Sigma^\mathcal{N}}\ReachProb{u}{\MDPValSigmaN}{\fin}\\
&=\max_{\sigma\in\Sigma^\mathcal{N}} \sum_{i=1}^n \RewardVal(t_i) \cdot \ReachProb{u}{\MDPValSigmaN}{t_i} && \text{using \cref{lem:essential_reach}}\\
&=\max_{\sigma\in\Sigma^\mathcal{N}} \sum_{i=1}^n \RewardOptVal(t_i) \cdot \ReachProb{u}{\MDPValSigmaN}{t_i} && \text{single choice for states in $T$}\\
&=\max_{\sigma\in\Sigma^\mathcal{M}} \sum_{i=1}^n \RewardOptVal(t_i) \cdot \ReachProb{u}{\MDPValSigma}{t_i} && \text{by construction of }\mathcal{N}\\
&=\max_{\sigma\in\Sigma^\mathcal{M}} \sum_{i=1}^n \rho(t_i)\cdot c \cdot \ReachProb{u}{\MDPValSigma}{t_i}\\
&=c \cdot\max_{\sigma\in\Sigma^\mathcal{M}} \sum_{i=1}^n \rho(t_i)\cdot  \ReachProb{u}{\MDPValSigma}{t_i}\\
&=c \cdot\max_{\sigma\in\Sigma^\mathcal{M}}\Reward{u}{\MDPValSigma}\\
&=c \cdot\RewardOptValM(u).\\
\end{align*}
Above, when we use \cref{lem:essential_reach}, the values of $S,U,T$ are $\mathcal{N}_\val^\sigma, T, \fin$ respectively.

For state-action pairs $(s,a)\in S_N$, the result directly follows from the above result for states the definition of rewards.

\section{Removing weights from trivially parametric MDPs}

\begin{definition}
Given a Markov chain $\mathcal{C}=(S,\mu,T,\rho)$, an initial state $s_0\in S$
and sets of states $A,B\subseteq S$, we denote by
$\ReachProbu{s_0}{\mathcal{C}}{A}{B}$ the probability of staying in $A$ until eventually reaching $B$ from $s_0$. If $s_0 \in B$ then $\ReachProbu{s_0}{\mathcal{C}}{A}{B} = 1$
otherwise:
  \[
    \ReachProbu{s_0}{\mathcal{C}}{A}{B} = \sum \multiset{
    \mathbb{P}_{\mu}(\pi) \mid \pi = s_0 \dots s_n \in A^+ B }.
  \]    
\end{definition}

When all runs from $s_0$ go through a set in $U$ before the run reaches a state in $T$, we can decompose the probability of reaching $T$ using the following lemma \cite{NWR_paper}.

\begin{lemma}[From \cite{NWR_paper}]\label{lem:essential_reach}
Let $\mathcal{C} = (S, \mu)$ be a Markov chain, sets of states $U,T\subseteq S$, and a state $s_0\in S \setminus U$. If $\ReachProbu{s_0}{\mathcal{C}}{(S \setminus U)}{T}=0$, then:
\begin{equation*}
    \ReachProb{s_0}{C}{T}=\sum_{u\in U}\ReachProbu{s_0}{C}{(S \setminus U)}{u}\cdot \ReachProb{u}{C}{T}.
\end{equation*}
\end{lemma}

Let $\mathcal{M} = (S, A, \Delta, T, \rho)$ be a w\~pMDP with target states $T=\Set{ t_0, t_1, \cdots, t_n }$ such that $0=\rho(t_0)<\rho(t_1)<\cdots<\rho(t_n)$. We will now create a (non-weighted) \~pMDP $\mathcal{M'}=(S',A',\Delta', T')$, which has the same never-worse relations as the given w\~pMDP $\mathcal{M}$. We do so by adding a new action $a\not\in A$ which ensures that the ordering of weights of target states in $\mathcal{M}$ is the same as the ordering of rewards of those states in $\mathcal{N}$. Formally:
\[
\begin{array}{rl}
    S' &= S \uplus \Set{\fail,\fin}, \\
    A' &= A \uplus \Set{a}, \\
    (s,a',s')&\in\Delta' \text{ iff } (s,a',s')\in\Delta \text{ for $s,s'\in S$ and $a'\in A$},\\
    (t,a,\fin)&\in\Delta'\text{ for all }t\in\{t_1,\dots,t_n\},\\
    (t_i,a,t_{i-1})&\in\Delta'\text{ for all }i\in\{1,\dots,n\},\\
    (t_0,a,\fail)&\in\Delta',\\
    T'   &= \Set{\fail,\fin}.\\
\end{array}
\]
We use $a_i$ to denote the state-action pair $(t_i,a)$ for convenience. 

Without loss of generality we also assume that any state-action pair $v\in S_N$, that does not have a path to $\fin$, has an edge to $t_0$. This is so that we can avoid trapping loops with probability one. Given a  wpMDP $\mathcal{M}$, one can think of a valuation $\val\in\GraphPresVals$ as a family of full support probability distributions $\val=(\val_u\in\mathbb{D}(uE))_{u\in S_N}$, where $\mathbb{D}(U)$ is the set of all full support probability distributions over the set $U$. We will use the notation $\val^\mathcal{M}_u(s)$ to denote the probability on the edge from the state-action pair $u$ to the state $s$ induced by $\val$ in $\mathcal{M}$.

\begin{remark}
    In the case of (weighted or non-weighted) trivially parametric MDPs, the set of all valuations in $\GraphPresVals$ is exactly the set of all families of full support probability distributions $\mu=(\mu_u\in\mathbb{D}(uE))_{u\in S_N}$.
\end{remark}

\subsection{Proof of \cref{lem:strict-order-on-targets}}

It is easy to see that
$\ReachProb{t_0}{\MDPValSigmaN}{\fin}=\ReachProb{(t_0,a)}{\MDPValSigmaN}{\fin}=0$.
We can also see that
$\ReachProb{t_i}{\MDPValSigmaN}{\fin}=\ReachProb{(t_i,a)}{\MDPValSigmaN}{\fin}$
for $i\in\{1,\dots,n\}$. Now, we use \cref{lem:val_compare_child} and
the fact that $\RewardVal(\fin)=1$ is the maximum possible reward to deduce
that
$\ReachProb{t_{i-1}}{\MDPValSigmaN}{\fin}\leq\ReachProb{(t_i,a)}{\MDPValSigmaN}{\fin}\leq
1$ for $i\in\{1,\dots,n\}$. We will now argue that the inequalities are
in fact strcit: Suppose
$\ReachProb{(t_i,a)}{\MDPValSigmaN}{\fin}=1$ for some $i\in\{1,\dots,n\}$.
This would imply, by \cref{lem:bellman_eq}, that
$\ReachProb{t_{i-1}}{\MDPValSigmaN}{\fin}=\ReachProb{(t_i,a)}{\MDPValSigmaN}{\fin}=\ReachProb{(t_{i-1},a)}{\MDPValSigmaN}{\fin}=1$.
By repeatedly applying this reasoning, we get that
$\ReachProb{t_0}{\MDPValSigmaN}{\fin}=\ReachProb{(t_0,a)}{\MDPValSigmaN}{\fin}=1$,
which is a contradiction. Hence, the inequality must be strict. This concludes
the proof. \qed

\subsection{Proof of \cref{thm:w_to_b_trivparam}}

We show that the property: the ordering of weights of the target states in
$\mathcal{M}$ is the same as the ordering of $\RewardOptValN$ of those states
in $\mathcal{N}$, suffices to preserve all the never-worse relations of
$\mathcal{M}$ into $\mathcal{N}$.
The proof is split into two parts: ``if'' and ``only if'' respectively.

\subsection*{First case: $v\trianglelefteq W$ in $\mathcal{M}$ if $v\trianglelefteq W$ in $\mathcal{N}$}
We prove the contrapositive: Suppose $v\not\trianglelefteq W$ in the w\~pMDP $\mathcal{M}$, which means there is a valuation $\val\in\GraphPresVals$ for which $\RewardOptValM(v)>\RewardOptValM(w)$ for all $w\in W$. We will now construct a valuation $\val'\in\GraphPresValsN$, so that $\RewardOptValprimeN(v) > \RewardOptValprimeN(w)$ for all $w\in W$. The existence of such a $\val'$ will show that $v\not\trianglelefteq W$ in $\mathcal{N}$, which concludes the proof in this direction.

By choosing some $0 < \varepsilon < \frac{1}{\rho(t_n)}$, one can easily construct probability distributions $\valprime^\mathcal{N}_{a_i}$ for $a_i \in \Set{ a_0,a_1,\cdots,a_n }$ so that $\RewardOptValprimeN(t_i) =\RewardOptValprimeN(a_i)= \rho(t_i)\cdot \varepsilon$. Once we have constructed this $\valprime_{a_i}^\mathcal{N}$ for $a_i \in \Set{ a_0,a_1,\cdots,a_n }$, we can then define the valuation $\val'$ such that for all $(u, v) \in (S_N \times S) \cap E$:

\[\valprime_u^\mathcal{N}(v) = \begin{cases}
\val_u^\mathcal{M}(v)   & \text{if}\;u \in S_N, \\
\valprime_{a_i}^\mathcal{N}(v) & \text{if}\;u = a_i \text{ for }0\leq i\leq n.
\end{cases}\]

We take into consideration the fact that in $\mathcal{N}$, with the family of distributions $\valprime$, we have $\RewardOptValprimeN(t_i)=\rho(t_i) \cdot \varepsilon$ for $1\leq i\leq n$. We will first show that for any $u\in V$, we have $\RewardOptValprimeN(u)=\varepsilon\cdot\RewardOptValM(u)$. To prove this, we first note that in $\mathcal{N}$, any path from $u$ to $\fin$ will always go through a vertex in $T$. 
\begin{align*}
\RewardOptValprimeN(u)&=\max_{\sigma \in \Sigma^\mathcal{N}}\ReachProb{u}{\mathcal{N}[\valprime]}{\fin}_\mathcal{N}\\
&=\max_{\sigma \in \Sigma^\mathcal{N}}\ReachProb{u}{\MDPValprimeSigmaN}{\fin}\\
&=\max_{\sigma \in \Sigma^\mathcal{N}} \sum_{i=1}^n \RewardValprimeN(t_i) \cdot \ReachProb{u}{\MDPValprimeSigmaN}{t_i} && \text{using \cref{lem:essential_reach}}\\
&=\max_{\sigma \in \Sigma^\mathcal{N}} \sum_{i=1}^n \RewardOptValprimeN(t_i) \cdot \ReachProb{u}{\MDPValprimeSigmaN}{t_i} && \text{single choice for targets}\\
&=\max_{\sigma \in \Sigma^\mathcal{M}} \sum_{i=1}^n \RewardOptValprimeN(t_i) \cdot \ReachProb{u}{\MDPValprimeSigma}{t_i} && \text{by construction of $\mathcal{N}$}\\
&=\max_{\sigma \in \Sigma^\mathcal{M}} \sum_{i=1}^n \RewardOptValprimeN(t_i) \cdot \ReachProb{u}{\MDPValSigma}{t_i} && \text{by construction of $\valprime$}\\
&=\max_{\sigma \in \Sigma^\mathcal{M}} \sum_{i=1}^n \rho(t_i) \cdot \varepsilon \cdot \ReachProb{u}{\MDPValSigma}{t_i}\\
&=\varepsilon \cdot \max_{\sigma \in \Sigma^\mathcal{M}} \sum_{i=1}^n \rho(t_i) \cdot \ReachProb{u}{\MDPValSigma}{t_i}\\
&=\varepsilon \cdot \max_{\sigma \in \Sigma^\mathcal{M}} \RewardValM(u)\\
&=\varepsilon \cdot \RewardOptValM(u).
\end{align*}

Above, when we use \cref{lem:essential_reach}, the values of $S,U,T$ are $\mathcal{N}_\valprime^\sigma, T, \fin$ respectively.

Thus, if $\RewardOptValM(v) > \RewardOptValM(w)$ for all $w\in W$ in $\mathcal{M}$, we have
\begin{align*}
    \RewardOptValprimeN(v)&=\varepsilon\cdot\RewardOptValM(v)\\
    &>\varepsilon\cdot\RewardOptValM(w)\\
    &=\RewardOptValprimeN(w).
\end{align*}

\subsection*{Second case: $v\trianglelefteq W$ in $\mathcal{M}$ only if $v\trianglelefteq W$ in $\mathcal{N}$}

Again, we prove the contrapositive: Suppose $v\not\trianglelefteq W$ in the (non-weighted) \~pMDP $\mathcal{N}$. This means that there exists a family of distributions $\valprime$ such that $\RewardOptValprimeN(v)>\RewardOptValprimeN(w)$ for all $w\in W$.

Let $K=\Set{ k_0=0,k_1,\cdots,k_m }$ be the set of all the values obtained by all the vertices in $\mathcal{N}$ in increasing order (i.e. $k\in K$ if, and only if $\RewardOptValprimeN(u)=k$ for some $u\in V'$). We have $k_0=0$ since $\RewardOptValprimeN(fail)=0$. Let $\RewardOptValprimeN(v)=k_i$ where $1\leq i\leq m$ (where $v$ is the vertex in our assumption for the contrapositive). We partition $V'$ into $\Set{V'_0,V'_1,\cdots,V'_i}$ such that all the vertices with value $k_j<k_i$ are in a respective partition $V'_j$ and those with value $k_j\geq k_i$ are in $V'_i$:
\[ \begin{array}{rl}
    V'_i &= \Set{u \in V' | \RewardOptValprimeN(u) \geq k_i}, \\
    V'_{j < i} &= \Set{u \in V' | \RewardOptValprimeN(u) = k_j}.
\end{array} \]
This partition is illustrated in the following \cref{fig:partition}.

\begin{figure}[ht]
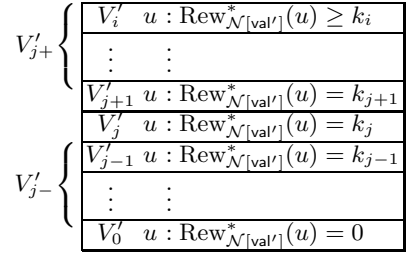

    \centering        \begin{tabular}{r@{\,}|cr@{\hskip3pt}l|}
        \cline{2-4}
        \ldelim\{{3}{*}[$V'_{j+}$] & $V'_i$ & $u$&$: \RewardOptValprimeN(u) \geq k_i$ \\
        \cline{2-4}
         & $\vdots$ & & $\vdots$\\
        \cline{2-4}
         & $V'_{j+1}$ & $u$&$: \RewardOptValprimeN(u) = k_{j+1}$\\
        \cline{2-4}
         & $V'_j$ & $u$&$: \RewardOptValprimeN(u) = k_{j}$ \\
        \cline{2-4}
        \ldelim\{{3}{*}[$V'_{j-}$] & $V'_{j-1}$ & $u$ & $:\RewardOptValprimeN(u) = k_{j-1}$ \\
        \cline{2-4}
         & $\vdots$ & & $\vdots$ \\
        \cline{2-4}
         & $V'_0$ & $u$ & $:\RewardOptValprimeN(u) = 0$ \\
        \cline{2-4}
    \end{tabular}
    \caption{Distribution of the nodes according to their values.}
    \label{fig:partition}
\end{figure}

We call each component of this partition a ``layer'' where $V'_0$ is the lowest layer and $V'_i$ is the topmost layer. For the sake of convenience, for $0\leq j\leq i$, let us define $V'_{j-}:=V'_0\cup V'_1\cup\cdots\cup V'_{j-1}$ and $V'_{j+}:=V'_{j+1}\cup V'_{j+2}\cup\cdots\cup V'_i$. We now state a few properties of this partition:
\begin{enumerate}
    \item \label[property]{prope:1} $v\in V'_{i}$ and $W\subseteq V'_{i-}$ (i.e. $v$ is in the top layer and no vertex in $W$ is in that layer).
    \item \label[property]{prope:2} $\Set{t_0,\cdots,t_p }\subseteq V'_{i-}$ and $\Set{t_{p+1},\cdots,t_n } \subseteq V'_{i}$ where $p$ is the unique number for which $\RewardOptValprimeN(t_p) < \RewardOptValprimeN(v) \leq \RewardOptValprimeN(t_{p+1})$.
    \item \label[property]{prope:3} $|V_j\cap T|\leq1$ for $0\leq j<i$ (i.e. each layer apart from $V_i$ can have at most one state in $T$).
    \item \label[property]{prope:4} For all $0\leq j<i$, for all $u\in V'_j\cap S'$, we have $uE\cap V'_{j+}=\emptyset$ (i.e. no state-action pair can be in a layer above the corresponding state).
    \item \label[property]{prope:5} For all $0\leq j<i$, for all $u\in V'_{j}\cap S_N'$, $uE\cap V'_{j+}\neq\emptyset$, iff $uE\cap V'_{j-}\neq\emptyset$ (i.e. a state-action pair can have an edge to a state in a layer above it iff it has one to a state in a layer below it).
    \item \label[property]{prope:6} For all $0\leq j\leq i$, for all $u\in V'_{j}\cap (S'\setminus T')$, we have $uE\cap V'_{j}\neq\emptyset$ (i.e. every non-target state must have an edge to a vertex in the same layer).
\end{enumerate}
The first property follows from the fact that $v$ is the vertex with the smallest value in the top layer and its value is strictly greater than that of each vertex in $W$. The second property partitions the states in $T$ (i.e. the target states in $\mathcal{M}$) into those in the topmost layer $V'_i$ and those not in it. The third property follows from the assumption that no two states in $T$ have the same weight. The final three follow from \cref{lem:val_compare_child}.

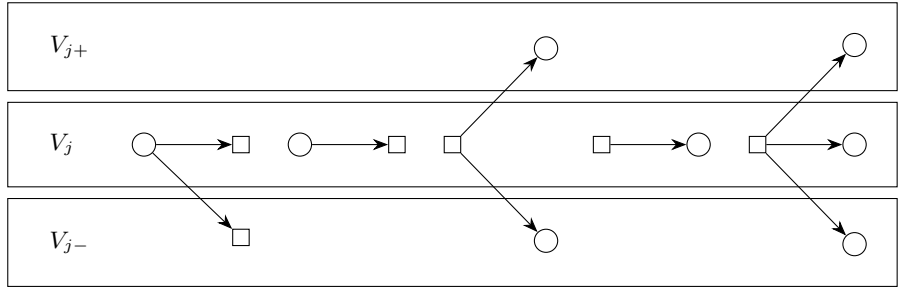
\begin{figure}[ht]
    \centering
    \begin{tikzpicture}[->]
            \tikzstyle{protagonist} = [shape=circle, draw]
            \tikzstyle{nature} = [shape=rectangle, draw]
            
            \node[protagonist] (case1_p1) {};
            \node[nature, right = of case1_p1] (case1_n1) {};
            \node[nature, below = of case1_n1] (case1_n2) {};
            \node[left = of case1_p1] (case1_invis1) {};
            \node[above = of case1_invis1] (case1_invis2) {};
            \node[below = of case1_invis1] (case1_invis3) {};
            
            \node[protagonist, right =.5 of case1_n1] (case2_p1) {};
            \node[nature, right = of case2_p1] (case2_n1) {};
            
            \node[nature, right =.5 of case2_n1] (case3_n1) {};
            \node[right = of case3_n1] (case3_invis) {};
            \node[protagonist, above = of case3_invis] (case3_p1) {};
            \node[protagonist, below = of case3_invis] (case3_p2) {};
            
            \node[nature, right =.5 of case3_invis] (case4_n1) {};
            \node[protagonist, right = of case4_n1] (case4_p1) {}; 
            
            \node[nature, right =.5 of case4_p1] (case5_n1) {};
            \node[protagonist, right = of case5_n1] (case5_p1) {};
            \node[protagonist, above = of case5_p1] (case5_p2) {};
            \node[protagonist, below = of case5_p1] (case5_p3) {};
            
            \node [draw, fit= (case1_invis2) (case3_p1) (case5_p2), inner sep=0.4cm, label={[xshift=1.2cm]left:$V_{j+}$}] (v_j_plus) {};
            
            \node [draw, fit= (case1_invis1) (case5_p1), inner sep=0.4cm, label={[xshift=1cm]left:$V_j$}] (v_j) {};
            
            \node [draw, fit= (case1_invis3) (case5_p3), inner sep=0.4cm, label={[xshift=1.2cm]left:$V_{j-}$}] (v_j_minus) {};
            
            \draw (case1_p1) edge [above] node {} (case1_n1);
            \draw (case1_p1) edge [above] node {} (case1_n2);
            
            \draw (case2_p1) edge [above] node {} (case2_n1);
            
            \draw (case3_n1) edge [above] node {} (case3_p1);
            \draw (case3_n1) edge [above] node {} (case3_p2);
            
            \draw (case4_n1) edge [above] node {} (case4_p1);
            
            \draw (case5_n1) edge [above] node {} (case5_p1);
            \draw (case5_n1) edge [above] node {} (case5_p2);
            \draw (case5_n1) edge [above] node {} (case5_p3);
        \end{tikzpicture}
    \caption{The only possibilities of edges enforced by the last 3 rules}
    \label{fig:partition_rules456}
\end{figure}

\begin{example}
\cref{fig:partition_rules456} shows the only possibilities of edges between layers that are enforced by the last three properties.
\end{example}

This partition on $\mathcal{N}$ induces one in $\mathcal{M}$. We call this partition $\Set{V_0,V_1,\cdots,V_i}$, where $V_j=V_j'\cap V$ and $V_{j+}$ and $V_{j-}$ are defined similarly for all $0\leq j\leq i$. Intuitively, we remove the new vertices that were added in the construction of $\mathcal{N}$. Note that some of the partitions might be empty after removing the vertices from $V'\setminus V$. 

All the six properties we stated above for the partition of $V'$ will also hold for this partition of $V$. This is because no vertices other than those in $T=\Set{ t_1,t_2,\cdots,t_n }$ had edges to the (new) vertices in $V'\setminus V$ that were removed and the vertices in $T$ have no outgoing edges in $\mathcal{M}$.

\begin{lemma}[From \cite{NWR_paper}]\label{lem:non-dec-path}
In a wpMDP $\mathcal{M} = (S,A,X,\delta,T,\rho)$, for any vertex $v\in V$ and a valuation $\val\in\GraphPresVals$, if a vertex $v\in V$ can reach some target state, there will be a path  from $v$ to a target vertex in which the values of $\RewardOptVal$ along the path are non-decreasing.
\end{lemma}

By \cref{lem:non-dec-path} and the fact that $\RewardOptValprimeN(v)=k_i>0$, each vertex in $V'_i$ will have a path to $\fin$ that stays in $V'_i$. Since each path from a vertex in $V$ to $\fin$ must go through a vertex in $T$, we get that each vertex in $V'_i$ will have a path to some vertex in $T$ that stays in $V'_i$.

To show that $v\not\trianglelefteq W$ in $\mathcal{M}$, we create a valuation (i.e., a family of distributions) $\val=(\val_u\in\mathbb{D}(uE))_{u\in S_N}$ for which $\RewardOptValM(v)>\RewardOptValM(w)$ for all $w\in W$. 

Let $0 < \varepsilon < 1$ be small enough and $m$ be the number of vertices in $V$.
Let $O:v_1,v_2,\cdots,v_m$ be an ordering on these vertices that preserves the ordering in the partition (vertices in lower layers occur earlier in the ordering than those in the higher layers). Formally, $v_{i_1}\in V_{j_1}$ and $v_{i_2}\in V_{j_2}$ with $j_1<j_2$ implies $i_1<i_2$.

Let $u\in V_i\cap S_N$ (i.e. a state-action pair in the top layer). We look at the shortest paths from $u$ to a target state in $T$ that stay in $V_i$. Let us call this set of paths $P_u^T$. Let us define the length of all paths in $P_u^T$ to be $\ell_u$. We define $\val^\mathcal{M}_u(p)$ for such vertices $u \in V_i\cap S_N$ and $p \in uE$ as follows:
\[
\val^\mathcal{M}_u(p)=\begin{cases}
1-(|uE|-1)\varepsilon & \text{if $p$ is a successor of $u$ in a path in $P_u^T$ and $p$ is}\\
& \text{the first one in $O$ among all such vertices}\\
\varepsilon & \text{else if $p\in uE$.}
\end{cases}
\]

Informally, we can say that $\nu_u$ ensures that one shortest path from $u$ to some target vertex is taken with probability almost 1. Recall that $m$ is the number of vertices in $V$, $\Set{t_{p+1},\cdots,t_n}$ are the target vertices in the topmost layer ($V_i$) and $\Set{t_0,\cdots,t_{p}}$ are the target vertices in $V_{i-}$.

\begin{lemma}
\label{lem:lower-bound-on-top-layer}
The (partial) family of distributions $\nu$ defined above for vertices in $V_i$ ensures that for any vertex $u\in V_i$: \[\RewardOptValM(u) \geq (1 - m\varepsilon)^{\ell_u} \cdot \rho(t_{p+1}). \]  
\end{lemma}

\begin{proof}
We will prove this lemma for all $u\in V_i\cap S_N$ by induction on the length of paths $\ell_u$. The proof for vertices in $V_i\cap S$ will automatically follow from \cref{lem:val_compare_child}.

\textbf{Base case:} If $\ell_u=1$, then by definition there is a vertex $t\in T$ for which $\nu_u(t)=1-(|uE|-1)\varepsilon$. 
\begin{align*}
    \RewardOptValM(u)&\geq(1-(|uE|-1)\varepsilon) \cdot \rho(t) \\
        &\geq(1-m\varepsilon) \cdot \rho(t) \\
        &\geq(1-m\varepsilon) \cdot \rho(t_{p+1}). && \text{by \cref{lem:strict-order-on-targets} and \cref{prope:2}}
\end{align*}

\textbf{Induction step:} Now, suppose the lemma holds for all $u\in V_i\cap S_N$ for which $\ell_u<\ell$.
Let $\ell_u=\ell$ and let $u,u_1,u_2,\cdots,u_{\ell-1},t$ be a path in $P_u^T$, where $\nu_u(u_1)=1-(|uE|-1)\varepsilon$. We know that this path is in $V_i$, thus $u_2\in V_i\cap S_N$ and $\ell_{u_2}=\ell-2$. Hence, we have $\RewardNu(u_2)\geq(1-m\varepsilon)^{\ell-2}\cdot\rho(t_{p+1})$ so that
\begin{align*}
    \RewardOptValM(u)&=(1-(|uE|-1)\varepsilon) \cdot \RewardOptValM(u_1)\\
    &\geq(1-m\varepsilon) \cdot \RewardOptValM(u_1) \\
    &\geq(1-m\varepsilon) \cdot \RewardOptValM(u_2)&&\text{by \cref{lem:val_compare_child}}\\
    &\geq(1-m\varepsilon) \cdot (1-m\varepsilon)^{\ell-2}\cdot\rho(t_{p+1}) \\
    &=(1-m\varepsilon)^{\ell-1} \cdot \rho(t_{p+1})\\
    &>(1-m\varepsilon)^\ell \cdot \rho(t_{p+1}).
\end{align*}
This concludes the proof of the lemma.\qed
\end{proof}

The next lemma gives an upper bound on all vertices of a subset if a bound on the rewards of the vertices having outgoing edges outside the subset and the weights of the target states inside the subset is already known.
\begin{lemma}
\label{lem:leave-set}
In a wpMDP $\mathcal{M} = (S,A,X,\delta,T,\rho)$, for any subset of vertices $U\subseteq V$, if the values of all vertices $u\in U$ with $uE \setminus U\neq\emptyset$ and all vertices in $U\cap T$ is at most $x$, then the value of all vertices in $U$ is at most $x$.
\end{lemma}

\begin{proof}
We first describe the intuition: Since each vertex in the set $U$ with an edge leaving $U$ has value at most $x$, we can modify the wpMDP so that all such vertices have value $x$ and are target states (with no outgoing edges). It is easy to see that the values of all the vertices in this modified wpMDP will be at least the values of those vertices in the original wpMDP. Now, each vertex in $U$ in this modified wpMDP will have weight at most $x$, since all the target states have values at most $x$ and there are no edges leaving $U$. Formally:

We know, using \cref{lem:bellman_eq}, that the following system of linear equations has a unique solution that describes the $\RewardOptVal$ values of the vertices of the wpMDP:
\[\RewardOptVal(v) = \begin{cases}
    \rho(v) & \text{if}\; v \in T \\
    0         & \text{if}\; v \in Z \\
    \mathrm{max}_{u \in vE} \, \RewardOptVal(u) & \text{if}\; v \in S \setminus (T \cup Z) \\
    \sum_{u \in vE} \delta(v, u)[\mathsf{val}] \cdot \RewardOptVal(u)  & \text{if}\; v \in S_N \setminus Z.
\end{cases}\]
where $Z$ is the set of all vertices which do not have a path to $T$. Since we have the upper bound of $x$ on the $\RewardOptVal$ values for all the vertices in $U$ that have an edge leaving $U$, we now consider the following system of equations for the vertices in $U$:
\[\RewardOptVal(v) = \begin{cases}
    x & \text{if}\; v \in (T \cap U) \text{ or }vE\setminus U\neq\emptyset \\
    0         & \text{if}\; v \in Z \\
    \mathrm{max}_{u \in vE} \, \RewardOptVal(u) & \text{if}\; v \in (S\cap U) \setminus (T \cup Z) \\
    \sum_{u \in vE} \delta(v, u)[\mathsf{val}] \cdot \RewardOptVal(u)  & \text{if}\; v \in (S_N\cap U) \setminus Z.
\end{cases}\]
If we make all the vertices in $( T\cap U) \cup \set {u\in U | uE\setminus U\neq\emptyset}$ target states with weight $x$, from \cref{lem:bellman_eq}, the values of the vertices in this modified wpMDP is the unique solution to the above system of equations (since no vertex in $U$ will now have a path outside $U$). Thus, we conclude that the upper bound on the vertices that can leave $U$ and the target states in $U$ implies an upper bound on all the vertices in $U$.
\end{proof}\qed

We will now define $\val$ for the vertices in $V_{i-}$. Let $u\in V_j\cap S_N$ for some $0\leq j<i$. We define the distribution for edges going out of $u$ as follows:
\[\val^\mathcal{M}_u(p)=
\begin{cases}
\frac{1}{|uE|}      & \text{if $uE\cap V_{j-}=\emptyset$ and $p\in uE$,}\\
1-(|uE|-1)\varepsilon & \text{else, if $p\in uE\cap V_{j-}$ and among all such}\\
                    & \text{vertices, $p$ is the first one in the ordering $O$,}\\
\varepsilon         & \text{else, if $p\in uE$.}
\end{cases}\]
Informally, we say that for $0<\varepsilon<1$ small enough, $\nu_u$ ensures that in any run, the next vertex will be in a lower layer with probability almost 1. This leads to the following lemma.

\begin{lemma}
\label{lem:upper-bound-on-lower-layers}
For all the vertices $u\in V_j$ with $0\leq j<i$, we have:
\[\RewardOptValM(u)\leq \rho(t_p)+jm\varepsilon \rho(t_n).\]
\end{lemma}

\begin{proof}
We prove the result by induction on the layer of the vertex $u \in V_j$. During induction, we will first prove the result for the vertices in $S_N$, and then we will use \cref{prope:4} of the partition and \cref{lem:leave-set} to show that the bound also applies to the remaining vertices in $S$.

\textbf{Base case:} Let $u\in V_0$, $\RewardOptValM(u)=0 \leq \rho(t_p)$ since $\rho(t_0)=0$ and $T\cap V_0=\Set{t_0}$ (by \cref{lem:strict-order-on-targets} and the way we construct the partition) and the fact that there are no edges going out of the partition $V_i$ (using \cref{prope:4,prope:5} of the partition) since there are no layers below this one.

Now, let us assume that we have $\RewardOptVal(u)\leq \rho(t_p)+(j-1)m\varepsilon \rho(t_n)$ for $u\in V_{j-}$.\\

\textbf{Induction step:} Let $u\in V_j\cap S_N$. If $uE\cap V_{j+}\neq\emptyset$, there exists (by \cref{prope:5} of the partition) a vertex $q$ in $V_{j-}$ for which $\nu_u(q)=1-(|uE|-1)\varepsilon$. We thus have
\begin{align*}
    & \RewardOptValM(u) \\
    &=(1-(|uE|-1)\varepsilon)\RewardOptValM(q)\\
    &{} +\varepsilon\sum_{v\in uE \setminus \Set{p}}\RewardOptValM(v)\\ 
    &\leq \RewardOptValM(q)+\varepsilon\sum_{v\in uE \setminus \Set{p}}\RewardOptValM(v)\\
    &\leq \RewardOptValM(q)+(|uE|-1)\varepsilon \rho(t_n)&& \text{reward cannot exceed } \rho(t_n)\\
    &\leq \RewardOptValM(q)+m\varepsilon \rho(t_n)\\
    &\leq (\rho(t_p)+(j-1)m\varepsilon n)+m\varepsilon \rho(t_n)&& \text{by Induction Hypothesis}\\
    &= \RewardOptValM(t_p)+jm\varepsilon \rho(t_n).
\end{align*}

We look at the subset of vertices $V_{j-}\cup V_j$. In this subset, we have the upper bound $\RewardNu(t_p)+jm\varepsilon \rho(t_n)$ on all vertices in $V_{j-}$ (by Induction Hypothesis) and on all vertices $u\in V_j$ with $uE\cap V_{j+}\neq\emptyset$. In short, we have the bound on all vertices that can leave the set $V_{j-}\cup V_j$. We now use \cref{lem:leave-set} to conclude that all the vertices in $V_j$ must have values bounded by $\RewardNu(t_p)+jm\varepsilon \rho(t_n)$.
\qed
\end{proof}

We will now conclude our proof of \cref{thm:w_to_b_trivparam}. We know that
$\rho(t_p)<\rho(t_{p+1})$. Using this, together with
\cref{lem:lower-bound-on-top-layer,lem:upper-bound-on-lower-layers}, we can
choose $\varepsilon $ small enough so that for all $u\in V_{i-}$ and all
$u'\in V_i$, we have $\RewardOptValM(u)<\RewardOptValM(u')$. \qed

This reduction from W-TAs to non-weighted TAs gives us the following result.
Note it is slightly stronger than what is claimed in the paper. Indeed, it
essentially says that the NWR is hard even for weighted parametric Markov
chains.
\begin{corollary}
\label{cor:w_nwr_conp_complete}
Given a w\~pMDP $\mathcal{M} = (S,A,\Delta,T,\rho)$, a non-empty vertex set $W \subseteq V$ , and a vertex $v \in V$ , determining whether $v\trianglelefteq W$ is \textCoNP-complete even if $|uE| = 1$ for all $u \in S$.
\end{corollary}

\begin{proof}
Follows directly from Le Roux and P\'{e}rez \cite[Theorem 4]{NWR_paper},
\cref{thm:w_to_b_trivparam} and the fact that in \cref{thm:w_to_b_trivparam},
if $uE\leq 1$ for $u \in S$ that is not a target vertex in
$\mathcal{M}$, then $uE\leq 1$ for such vertices in $\mathcal{N}$ too.\qed
\end{proof}

\section{Proof of \cref{thm:pNWR_coETR_complete}}

We first show a simple conversion from target arenas \cite{NWR_paper} to \~pMDPs so that the proof of \textCoNP completeness still works.

\subsection{From target arenas to trivially parametric MDPs}
A target arena will not be a \~pMDP only if there are multiple player vertices in the target arena which have an edge to the same nature vertex (one can think of player vertices as states and nature vertices as actions). This is not possible in MDPs since each action has exactly one incoming edge. Suppose there are player vertices $s_1,\dots,s_n$ all having an edge to a nature vertex $a$. We will now create an equivalent target arena with a new player vertex $s_a$ having only one outgoing edge to $a$ and $n$ new nature vertices $a_1,\dots,a_n$ such that $a_i$ has an edge from $a_i$ and an edge to $s_a$ for $0<i\leq n$. This is described in \cref{fig:TA_to_tpMDP}.

\begin{figure}[htbp]
\centering
\subfloat[A nature vertex with multiple incoming edges.\label{fig:TA_to_tpMDP_a}]{
\begin{tikzpicture}[shorten >=1pt,auto,node distance=.9 cm, scale = 0.7, transform shape]
        \tikzstyle{action} = [shape=rectangle, draw]

        \node[state](s1){$s_1$};
        \node[state](s2)[right=of s1]{$s_2$};
        \node[state](s3)[right=of s2]{$s_3$};
        \node[](dots)[right=of s3]{$\dots$};
        \node[state](sn)[right=of dots]{$s_n$};
        \node[action](a)[below=of s3]{$a$};
                
        \path[->] 
        (s1)   edge [above] node {} (a)
        (s2)   edge [above] node {} (a)
        (s3)   edge [above] node {} (a)
        (sn)   edge [above] node {} (a)
        ;
    \end{tikzpicture}
}
\subfloat[The modified arena.\label{fig:TA_to_tpMDP_b}] {
\begin{tikzpicture}[shorten >=1pt,auto,node distance=.9 cm, scale = 0.7, transform shape]
        \tikzstyle{action} = [shape=rectangle, draw]

        \node[state](s1){$s_1$};
        \node[state](s2)[right=of s1]{$s_2$};
        \node[state](s3)[right=of s2]{$s_3$};
        \node[](dots)[right=of s3]{$\dots$};
        \node[state](sn)[right=of dots]{$s_n$};
        \node[action](a1)[below=of s1]{$a_1$};
        \node[action](a2)[below=of s2]{$a_2$};
        \node[action](a3)[below=of s3]{$a_3$};
        \node[action](an)[below=of sn]{$a_n$};
        \node[state](sa)[below=of a3]{$s_a$};
        \node[action](a)[below=of sa]{$a$};
                
        \path[->] 
        (s1)   edge [above] node {} (a1)
        (s2)   edge [above] node {} (a2)
        (s3)   edge [above] node {} (a3)
        (sn)   edge [above] node {} (an)
        (a1)   edge [above] node {} (sa)
        (a2)   edge [above] node {} (sa)
        (a3)   edge [above] node {} (sa)
        (an)   edge [above] node {} (sa)
        (sa)   edge [above] node {} (a)
        ;
    \end{tikzpicture}
}
\caption{A description of how we convert a target arena with actions having multiple incoming edges into one where each action has exactly one incoming edge.} \label{fig:TA_to_tpMDP}
\end{figure}
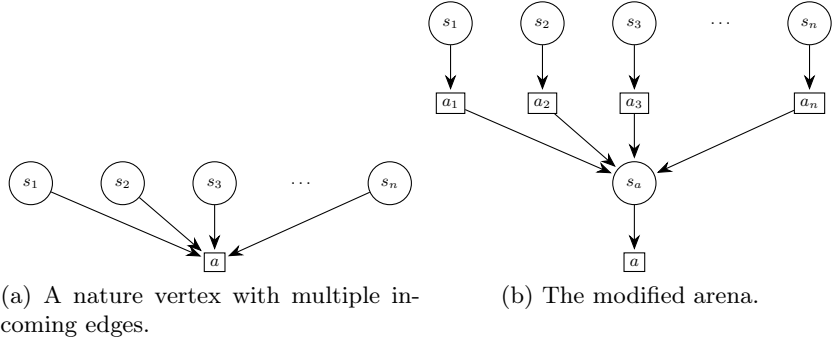

It is left to the reader to verify that this is indeed a \~pMDP and none of the reachability probabilities change from the construction described above.

Now, we move to proving \cref{thm:pNWR_coETR_complete}. We split the proof into two parts. First, we show that given a wpMDP
$\mathcal{M}=(S,A,X,\delta,T,\rho)$, a vertex $v\in V$, a non-empty subset
$W\subseteq V$, deciding whether $v\trianglelefteq W$ is \textCoETR-hard.
Next, we will give an \textETR{} formula equivalent to asking the complement of
the aforementioned question.

\subsection{Deciding the NWR is \textCoETR-hard}
\begin{definition}[\textsc{BCon4Ineq}] 
The \textnormal{bounded-conjunction-of-inequalities} problem asks: given a family of polynomials $f_1,\cdots,f_m$, of degree 4, does there exist some valuation $\mathsf{val}:X\rightarrow (0,1)$ such that $\bigwedge_{i=0}^mf_i(\mathsf{val})<0$?
\end{definition}

It is known that \textsc{BCon4Ineq} is \textETR-hard \cite[Lemma 5]{junges2021}.

To show \textCoETR-hardness, we first reduce \textsc{BCon4Ineq} to the problem of deciding whether there exists a graph-preserving valuation for a given $\mathcal{M}$. Let $f_1,\cdots,f_m$ be the polynomials. Let $k$ be a constant such that $k>\max_{1\leq i\leq m}\set{k_i}$ where $k_i$ is the sum of the absolute values of the coefficients of the polynomial $f_i$. We define new polynomials $f_1',\dots,f_m'$ where $f_i'=\frac{f_i}{k}$. Note that for any $\val:X\rightarrow (0,1)$, we have $\bigwedge_{i=0}^mf_i[\val]<0$ if and only if $\bigwedge_{i=0}^mf'_i[\val]<0$ and we have $-1<f'_i[\val]<1$ for any such $\val$. Let $X=\set{x_1,\dots,x_n}$ be the set of all variables occurring in $f_i$ for $1\leq i\leq m$.

We now define the pMDP $\mathcal{M}=(S,A,X,\delta,T)$ as follows:
\[
\begin{array}{rl}
    S &= \set{s1,\dots,s_m}\cup\set{s_1',\dots,s_n'}\cup\set{\fail,\fin}, \\
    A &= \set{n_1,\dots,n_m}\cup\set{v_1,\dots,v_n}, \\
    X &= \set{x_1,\dots,x_n}, \\
    \delta(s_i,n_i,\fail) &= -f_i'\text{ for }1\leq i\leq m, \\
    \delta(s_i,n_i,\fin) &= 1+f_i'\text{ for }1\leq i\leq m, \\
    \delta(s_i',v_i,\fail) &= x_i\text{ for }1\leq i\leq n, \\
    \delta(s_i',v_i,\fin) &= 1-x_i\text{ for }1\leq i\leq n, \\
    T'   &= \Set{\fail,\fin}.\\
\end{array}
\]
\begin{figure}[ht]
    \centering
        \begin{tikzpicture}[shorten >=1pt,auto,node distance=3 cm, scale = 0.7, transform shape]
            \tikzstyle{protagonist} = [shape=circle, draw]
            \tikzstyle{nature} = [shape=rectangle, draw]
                
            \node[nature](n1){$n_1$};
            \node[nature](n2)[right=of n1]{$n_2$};
            \node[nature](n3)[right=of n2]{$n_3$};
            \node[](dots)[right=of n3] {$\cdots$};
            \node[nature](nm)[right=of dots]{$n_m$};
            \node[nature](x1)[below=5cm of n1]{$v_1$};
            \node[nature](x2)[below=5cm of n2]{$v_2$};
            \node[nature](x3)[below=5cm of n3]{$v_3$};
            \node[](duts)[right=of x3]{$\cdots$};
            \node[nature](xn)[below=5cm of nm]{$v_n$};
            \node[state, accepting](fail)[below=2cm of n2]{$\fail$};
            \node[state, accepting](fin)[below=2cm of dots]{$\fin$};
            
            \path[->] (n1) edge [left,bend right] node [align=center] {$-f_1'$} (fail)
            (n2) edge [left,bend right] node [align=center] {$-f_2'$} (fail)
            (n3) edge [above] node [left,pos=0.7] {$-f_3'$} (fail)
            (nm) edge [above] node [above,pos=0.2] {$-f_m'$} (fail)
            (n1) edge [above] node [above,pos=0.2] {$1+f_1'$} (fin)
            (n2) edge [above] node [above,pos=0.2] {$1+f_2'$} (fin)
            (n3) edge [above] node [right,pos=0.2] {$1+f_3'$} (fin)
            (nm) edge [right,bend left] node [align=center] {$1+f_m'$} (fin)
            
            (x1) edge [left,bend left] node [align=center] {$x_1$} (fail)
            (x2) edge [left,bend left] node [align=center] {$x_2$} (fail)
            (x3) edge [above] node [left,pos=0.7] {$x_3$} (fail)
            (xn) edge [above] node [above,pos=0.2] {$x_n$} (fail)
            (x1) edge [above] node [below,pos=0.2] {$1-x_1$} (fin)
            (x2) edge [above] node [below,pos=0.2] {$1-x_2$} (fin)
            (x3) edge [above] node [right,pos=0.2] {$1-x_3$} (fin)
            (xn) edge [right,bend right] node [align=center] {$1-x_n$} (fin)
            ;
        \end{tikzpicture}
    \caption{The pMDP used to show \textCoETR-hardness of determining whether
    there exists a graph-preserving valuation for a pMDP.}
    \label{fig:coETR_pTA_appendix}
\end{figure}
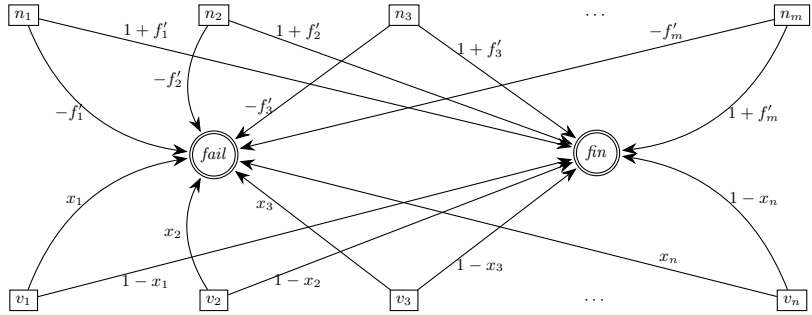
All the edges from the set $\Set{ n_1, \dots, n_m }$ to $\fail$ ensure that $0
< -f_i' < 1$ for $1 \leq i \leq m$ and the edges to $\fin$ ensure that
$0<1+f_i'<1$. Both of these give us that $-1<f_i'<1$. The edges from the set
$\Set{ v_1, \dots, v_n }$ to $\Set{ \fail, \fin }$ ensure that the variables
can only take values from the open set $(0,1)$. Thus, $\mathcal{M}$ has a
graph-preserving valuation if and only if there is some $\val : X \rightarrow
(0,1)$ for which $0 <  -f_i', 1+f_i' < 1$. But note that the range of $f_i$ is always
included in the interval $(-1,1)$ by our construction of $f_i'$ for
$1\leq i\leq m$. Thus, $\mathcal{M}$ has a graph-preserving valuation if and
only if there is some $\val : X \rightarrow (0,1)$ such that $f_i' < 0$, which
in turn holds if and only if $f_i<0$. To conclude, we observe that $\fin
\not\trianglelefteq \fail$ if and only if there is a graph-preserving
valuation for the pMDP we constructed. This concludes the proof. \qed

\subsection{Deciding the NWR is in \textCoETR}
The length of an ETR formula is defined as the number of symbols used to write
it. We assume binary encoding of the coefficients in $\delta$ and in the
constants in ETR formulas. Further, we write $|\delta|$ to denote the maximal
representation size of a polynomial in its image: the sum of $\log_2 c$ for
all of its coefficients $c$.

\begin{lemma}
\label{lem:copNWR_in_ETR}
Consider a wpMDP $\mathcal{M}=(S,A,X,\delta,T,\rho)$, a vertex $v\in V$ and a
subset of vertices $W\subseteq V$. We construct, in polynomial time, an ETR
sentence $\Phi_{\lnot \text{wpNWR}}$ of length $O(|S|^2|A|^2|\delta|)$ such that $\varphi$ is true
if and only if there is no graph preserving valuation $\mathsf{val}\in\GraphPresVals$ such that
for all $w\in W$ we have $\RewardOptValM(v) > \RewardOptValM(w)$.
\end{lemma}

The idea is to encode the concept of rewards into logical formulas using the equivalent notation described in \cref{lem:bellman_eq}.

For the encoding, for all vertices $v\in V$, we have a variable $y_v$ along with
all the variables in $X$. We define $Y= \Set{y_v \mid v\in V}$. We partition $V$
into 4 sets $\Set{T,Z,P,N}$ where $T$ is the set of target states, $Z$ is the
set of vertices which have no path to any target state, $P=S\setminus(T\cup
Z)$ and $N=S_N\setminus Z$ (recall that $S_N$ is the set of all state-action
pairs). This partition can be computed in polynomial time.
We first encode the values of vertices in ETR.

First, we encode the values of the target vertices and the vertices that do not have paths to target states:
\[
\Phi_T=\bigwedge_{v\in T}(y_v=\rho(v)) \quad \text{and} \quad \Phi_Z=\bigwedge_{v\in Z}(y_v=0).
\]
Next, we encode the values of the states using the fact that the value of a state is exactly the maximum of the values of its children:
\[
\Phi_P=\bigwedge_{v\in S}\bigg(\Big(\bigwedge_{u\in vE}y_v\geq y_u\Big)\wedge\Big(\bigvee_{u\in vE}y_v=y_u\Big)\bigg).
\]
Now we encode the values of the state-action pairs which is a convex combination of the values of its children and the combination is determined by $\delta$:
\[
\Phi_N=\bigwedge_{v\in S_N}\Big(y_v=\sum_{u\in vE}\delta(v,u)\cdot y_u\Big).
\]

Thus, we obtain the following formula that encodes the rewards of all the vertices in $\mathcal{M}$:
\[
\Phi_{\RewardOptVal}=\Phi_T\wedge\Phi_Z\wedge\Phi_P\wedge\Phi_N.
\]
The reader can easily verify, using \cref{lem:bellman_eq}, that $\Phi_{\RewardOptVal}$ has a unique solution which assigns the optimal reward of the vertex $v$ to $y_v$ for any $v\in V$.

Now, we will encode the graph preserving condition, which says that no outgoing edge from a state-action pair can have probability 0 and the sum of all the probabilities on edges going out of a vertex in $S_N$ must be 1.

\[
\Phi_{gp}=\Big(\bigwedge_{\substack{v\in S_N \\ u\in vE}}\delta(v,u)>0\Big)\wedge\Big(\bigwedge_{v\in S_N}\sum_{u\in vE}\delta(v,u)=1 \Big).
\]

The last thing to encode is the complement of the never-worse
relation:
\[
\Phi_{v\not\trianglelefteq W}=\bigwedge_{w\in W}y_v>y_w.
\]

Now, the final encoding of the complement of the NWR in ETR will be:
\[
\Phi_{\lnot \text{wpNWR}}=\exists X\exists Y:\Phi_{\RewardOptVal} \wedge\Phi_{gp}\wedge\Phi_{v\not\trianglelefteq W}.
\]

This gives us an encoding into \ETR{} for deciding whether $W\subseteq V$ is \textit{not} never-worse than $v\in V$ which shows that checking the never-worse relation is in \textCoETR{}.

\section{Proof of \cref{thm:equivalence_complexity}}
We first show that deciding equivalences in pMDPs and \~pMDPs is \textCoETR-complete and \textCoNP-complete respectively.
\subsection{Reducing the NWR decision problem to the equivalence decision problem}
Let $\mathcal{M}=(S,A,X,\delta,T)$ be any pMDP, and $u,v\in S$ any two states in the pMDP. We will now construct another pMDP $\mathcal{M'}=(S',A',X,\delta',T)$ where $S'=S\cup\set{p,q}$ and $A'=A\cup\set{b,c}$ such that $p\sim q$ in $\mathcal{M'}$ iff $u\trianglelefteq v$ in $\mathcal{M}$. We do this by defining $\delta'$ as:
\[\delta'(s,a,s')=
\begin{cases}
\delta(s,a,s')      & \text{if $s,s'\in S$ and $a\in A$,}\\
1 & \text{if $s\in\set{p,q}$, $a=b$ and $s'=v$,}\\
1 & \text{if $s=p$, $a=c$ and $s'=u$,}\\
0 & \text{otherwise.}
\end{cases}\]
This is described in \cref{fig:MDP_nwr_to_equivalence}.
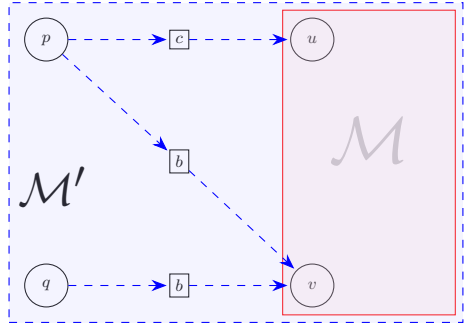
\begin{figure}[ht]
    \centering
    \begin{tikzpicture}[shorten >=1pt,auto,node distance=1.9 cm, scale = 0.7, transform shape]
    
        \node[state](p){$p$};
        \node[nature](pc)[right=of p]{$c$};
        \node[nature](pb)[below=of pc]{$b$};
        \node[nature](qb)[below=of pb]{$b$};
        \node[state](q)[left=of qb]{$q$};
        \node[state](u)[right=of pc]{$u$};
        \node[state](v)[right=of qb]{$v$};
        \node[](inv1)[right=of u]{};
        \node[](inv2)[right=of v]{};
        \node[](m')[above=1cm of q, text opacity=0.2]{\Huge{$\mathcal{M'}$}};
        \node       (X)    [draw=red, fit= (inv1) (v) (u), inner sep=0.1cm, fill=red!20, fill opacity=0.2] {\Huge{$\mathcal{M}$}};
        
        \node (Y) [draw=blue, fit= (inv1) (q) (p), inner sep=0.2cm, fill=blue!20, fill opacity=0.2, dashed] {};
        
        \path[->,color=blue,dashed] 
        (p)   edge [above] node [align=center] {} (pc)
              edge [above] node [align=center] {} (pb)
        (pc)  edge [above] node [align=center] {} (u)
        (pb)  edge [above] node [align=center] {} (v)
        (q)   edge [above] node [align=center] {} (qb)
        (qb)  edge [above] node [align=center] {} (v)
        ;
    \end{tikzpicture}
    \caption{The MDP $\mathcal{M'}$ where $p\sim q$ iff $u\trianglelefteq V$ in $\mathcal{M}.$}
    \label{fig:MDP_nwr_to_equivalence}
\end{figure}

The proof follows from the fact that $q\sim v$; and $v\sim p$ holds iff there is no valuation for which the reward of $u$ is greater than the reward of $v$, that is, $u\trianglelefteq v$.

\subsection{Finding equivalence classes in Markov chains is in \textbf{P}}

We first define a \textit{parametric Markov chain} $\mathcal{M}=(S,A,X,\delta,T)$ as pMDPs in which each state which is not a target has an edge to exactly one action. 

To simplify the notation, we will merge the states with their unique actions to get an equivalent model to pMCs. We will refer to these models as pMCs from here onward. That is, we now only have states with edges to other states and polynomials on these edges.

\begin{definition}[trivially parametric Markov chains]
    A \textnormal{trivially parametric Markov chain} $\mathcal{M}=(S,E,T)$ is a pMC where the variable on each edge is unique.
\end{definition}
Note that we omit $X$ the set of variables since the each edge has a unique variable, and hence, the only restriction on the probabilities is that the distribution induced on the outgoing edges from any state must be full support. So, we only care about the underlying graph of the chain, that is, the set of states $S$, the set of edges $E$ and the set of target states $T=\set{\fail,\fin}$.

Note that a similar reduction like the one shown for \~pMDPs will not work on Markov chains because the reduction uses states with multiple choices.

We provide a few preliminaries before giving the polynomial time algorithm for computing/deciding equivalences.
We will fix $\mathcal{M}=(S,E,T)$ to be the Markov chain we use throughout this section. Let the number of NWR equivalence classes in $\mathcal{M}$ be $n$. We have a few assumptions to make the proof simpler. Without loss of generality we assume that the only vertices equivalent to $\fin$ and $\fail$ are $\fin$ and $\fail$ respectively. Next, we assume that there are no self loops , since if a state $s$ has a self loop with probability $x$, we can multiply the probabilities on all the other outgoing edges from $s$ by $\frac{1}{1-x}$ (there must be at least one such edge since $s$ has a path to $\fin$) and the reader can easily verify that the new valuation is a graph preserving one and preserves all the reachability probabilities as the original one. We also assume that each state that is not a target has at least two outgoing edges, because if a state has no outgoing edges, then it does not have a path to $\fin$, and hence, it is collapsed with $\fail$ by the previous assumption; and if it has exactly one outgoing edge, it is equivalent to its unique successor and hence, we can collapse it with the unique successor. Note that the collapse might lead to new self loops but the process will terminate since there are only finitely many states.

We show that for any Markov chain, we can create an equivalent Markov chain such that in the new Markov chain, each state that is not a target has exactly 2 outgoing edges.
\begin{lemma}\label{lem:MC_2children}
    For any given Markov chain $\mathcal{M}$, we can create another Markov chain $\mathcal{N}=(S',E',T)$ such that $sE'=2$ for all $s\in S'\backslash T'$, $S\subseteq S'$, for all $\val\in\GraphPresVals$, there is $\val'\in\GraphPresValsN$ such that for any $s\in S$, $\ReachProbM{s}=\ReachProb{s}{\mathcal{N}^\valprime}{\fin}$ and vice versa.
\end{lemma}
\begin{proof}
    \cref{fig:MC_2edges} describes how we use new states to ensure that each state has exactly two successors.

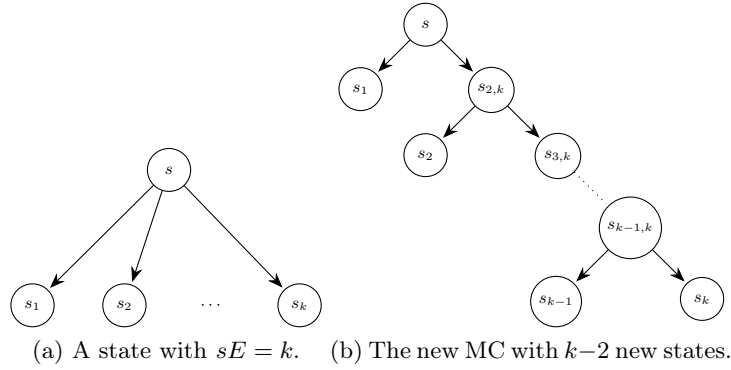
\begin{figure}[htbp]
\centering
\subfloat[A state with $sE=k$.\label{fig:MC_2edges_a}]{
\begin{tikzpicture}[shorten >=1pt,auto,node distance=.9 cm, scale = 0.7, transform shape]
        \tikzstyle{action} = [fill=black, shape=rectangle, draw]

        \node[state](s){$s$};
        \node[state](s1)[below left=2.8cm of s]{$s_1$};
        \node[state](s2)[right=of s1]{$s_2$};
        \node[](dots)[right=of s2]{$\cdots$};
        \node[state](sk)[right=of dots]{$s_k$};
                
        \path[->] 
        (s)   edge [above] node {} (s1)
                edge [above] node {} (s2)
                edge [above] node {} (sk)
        ;
    \end{tikzpicture}
}
\subfloat[The new MC with $k-2$ new states.\label{fig:MC_2edges_b}] {
\begin{tikzpicture}[shorten >=1pt,auto,node distance=.9 cm, scale = 0.7, transform shape]
        \tikzstyle{action} = [fill=black, shape=rectangle, draw]

        \node[state](s){$s$};
        \node[state](s1)[below left= of s]{$s_1$};
        \node[state](s2k)[below right=of s]{$s_{2,k}$};
        \node[state](s2)[below left= of s2k]{$s_2$};
        \node[state](s3k)[below right=of s2k]{$s_{3,k}$};
        \node[state](skminusonek)[below right=of s3k]{$s_{k-1,k}$};
        \node[state](skminusone)[below left=of skminusonek]{$s_{k-1}$};
        \node[state](sk)[below right=of skminusonek]{$s_k$};

        \path[->]
        (s) edge [above] node {} (s1)
            edge [above] node {} (s2k)
        (s2k) edge [above] node {} (s2)
            edge [above] node {} (s3k)
        (skminusonek) edge [above] node {} (skminusone)
            edge [above] node {} (sk)
         ;
        \draw[dotted] (s3k) to (skminusonek);
    \end{tikzpicture}
}
\caption{A description of how we convert a state with more than two outgoing edges to it having exactly two outgoing edges by adding $|sE|-2$ new states.} \label{fig:MC_2edges}
\end{figure}

Let $\val(s,s_i)=x_i$ for $1\leq i\leq k$. We have $x_i>0$ for all $i$ and $\sum_{1\leq i\leq k}s_i=1$ since $val\in\GraphPresVals$. We now define $\val'$ for all the edges in \cref{fig:MC_2edges_b}: $\val'(s_{i,k},s_i)=\frac{x_i}{1-\sum_{j<i}x_j}$ and $\val'(s_{i,k},s_{i+1},k)=\frac{1-\sum_{j\leq i}x_j}{1-\sum_{j<i}x_j}$ where we consider $s_{1,k}=s$ and $s_{k,k}=s_k$. Now the reader can easily verify  that the probabilities on the unique path from $s$ to $s_i$ is exactly $x_i$ for any $1\leq i\leq k$.

For the opposite direction, let $\val'\in\GraphPresValsN$. We define $\val(s,s_i)$ to be the product of the probabilities on the path $s,s_{2,k},\dots,s_{i,k},s_i$ induced by $\val'$. One can easily check that if $\val'$ was a graph preserving valuation, then so is $\val$.

Now, since we have shown that we can create valuations from a given valuation on either $\mathcal{M}$ or $\mathcal{N}$ that preserve the probabilities of reaching $s_i$ from $s$ in the subgraph, we can deduce that the probabilities of reaching $\fin$ from the new valuation must be the same as the old valuation.\qed 
\end{proof}

\begin{lemma}\label{lem:AS_implies_equivalence}
    If a vertex $s\in S$ can almost-surely reach an equivalence class $A$ then $s\in A$.
\end{lemma}
\begin{proof}
    Since $\mathcal{M}$ is a Markov chain, $s$ almost surely reaching $A$ implies that $A$ must be essential for $s$. Let $v\in A$ be any state in $A$ and $\val\in\GraphPresVals$. Using \cref{lem:essential_reach}, we get \begin{align*}
        \ReachProbM{s}&=\sum_{u\in A}\ReachProbu{s}{M^{\val}}{(S \setminus A)}{u}\cdot \ReachProbM{u},\\
        &=\ReachProbM{v}\cdot\sum_{u\in A}\ReachProbu{s}{M^{\val}}{(S \setminus A)}{u},\\
        &=\ReachProbM{v}.
    \end{align*}
    Since $\val$ was chosen arbitrarily, this implies that $s$ is equivalent to $v$ and hence, in the equivalence class $A$.\qed
\end{proof}

\begin{lemma}\label{lem:strict_bound_01}
    For all states $s\not\in T$ and all valuations $\val\in\GraphPresVals$, we have that $0<\ReachProbM{s}<1$.
\end{lemma}
\begin{proof}
    We have the assumption that $\fin$ and $\fail$ are the only states in their respective equivalence classes. This implies that each other vertex must have a path to $\fin$ (otherwise they would always have value $0$) and a path to $\fail$ (otherwise they would almost surely reach $\fin$ and, hence, in the same equivalence class). Since $\val$ is graph preserving, these paths to $\fin$ and $\fail$ must have nonzero probabilities and therefore $\ReachProbM{}$ can neither be $0$ nor $1$.\qed
\end{proof}

We will now give two proofs. The first one is a short and simple proof and the second one is a long but beautiful proof.

\subsubsection{First proof}

If we assume that any equivalence class has exactly one state that has edges going out from the equivalence class, then all the other states in the equivalence class must almost surely reach that state. This, along with the fact that almost surely reaching implies equivalence in states gives us an easy algorithm to find equivalence classes in a Markov chain by checking the almost sure reachability for all pairs of states in $S$.

A state $s$ is an \textit{exit} for a set $U\subseteq S$ if $s\in U$ and $sE\not\subseteq U$.

\begin{lemma}\label{lem:exit_children}
    If $u$ exits an equivalence class $U$, then both of the successors of $u$ must lie outside of $U$, that is, $uE\cap U=\emptyset$.
\end{lemma}
\begin{proof}
    We know that $u$ must have a successor outside of $U$, say $a$. We assume, for the sake of contradiction, that the other successor, say $b$,  is inside $U$. Now, for any given valuation $\val\in\GraphPresVals$, we have \begin{align*}
        \ReachProbM{u}&=\val(u,a)\ReachProbM{a}+\val(u,b)\ReachProbM{b}\\
        &=\val(u,a)\ReachProbM{a}+(1-\val(u,a))\ReachProbM{b}&\text{since }\val\in\GraphPresVals\text{ and }|uE|=2\\
        &=\val(u,a)\ReachProbM{a}+(1-\val(u,a))\ReachProbM{u}&\text{since }u\text{ and }b\text{ are equivalent}\\
        &=\ReachProbM{a}.
    \end{align*}
    Since $\val$ was chosen arbitrarily, we get that $a$ must be in $U$ which is a contradiction.\qed
\end{proof}

\begin{theorem}
        Any non-trivial equivalence class (i.e., one that does not contain $\fin$ or $\fail$) has exactly one exit.
\end{theorem}
\begin{proof}
    First, we observe that an equivalence class must have at least one exit since each state has a path to $\fin$. Now, we will show that it cannot have more than one exit. The proof is by contradiction. Suppose there is an equivalence class $U\subseteq S$ with at least two exits, say, $a$ and $b$. Since $a$ is an exit, there exist two states $c$ and $d$ such that the equivalence classes $\Tilde{c}$ and $\Tilde{d}$ are disjoint from each other and from $U$ from \cref{lem:AS_implies_equivalence} and \cref{lem:exit_children}. Since $c$ is not never-worse equivalent to $a$, there must be a valuation $\val\in\GraphPresVals$ such that $\ReachProbM{c}>\ReachProbM{a}$ or $\ReachProbM{c}<\ReachProbM{a}$. If the first is true, then $\ReachProbM{c}>\ReachProbM{a}>\ReachProbM{d}$ and if the second is true, then $\ReachProbM{c}<\ReachProbM{a}<\ReachProbM{d}$ since the reward value of $a$ is a convex combination of those of $c$ and $d$. Without loss of generality, assume $\ReachProbM{c}<\ReachProbM{a}<\ReachProbM{d}$. Let $m=|\set{\RewardOptValM{u}|u\in S}|$. We now construct a partition $\set{S_0,\dots,S_m}$ using $\val$ such that $\RewardOptValM{p}<\RewardOptValM{q}$ iff $p\in S_i$ and $q\in S_j$ with $i<j$. In this partition, similar to the one in the proof of \cref{thm:w_to_b_trivparam}, we refer to each $S_i$ as ``layers'' with $S_0=\set{\fail}$ being the bottom layer and $S_m=\set{\fin}$ being the top layer. Here are some properties of this partition:\begin{itemize}
        \item Any layer consists exactly of the states with the same probability of reaching $\fin$ with $\val$ with the layers with higher probability being higher than those with lower probabilities.
        \item Any equivalence class must be fully contained inside some layer.
        \item Each layer (other than the top and bottom one) must have at least one state that exits the layer (since each state must have a path to $\fin$).
        \item Each such exit must have one edge to a layer above and another edge to a layer below (since the probability of reaching $\fin$ from a state is a convex combination of those from its successors).
        \item From any state, there must be a path to a state exiting the layer such that the path is contained inside the layer. 
        \item From any state exiting a layer, there is a path to $\fin$ such that the second state in the path has strictly larger value and the values among the path are non-decreasing (by using the previous two properties).
        \item From any state exiting a layer, there is a path to $\fail$ such that the second state in the path has strictly smaller value and the values among the path are non-increasing (same reasoning as the previous property).
    \end{itemize}
    By our assumption, we can see that $\Tilde{c}$ is below the layer containing $U$ and $\Tilde{d}$ is above $U$. Let $S_i$ be the layer containing $U$\\
    Case 1: $b$ exits $S_i$ too. We observe that there are vertex disjoint paths from $a$ to $\fin$ and $b$ to $\fail$ by using the last two properties of the partition. Since these paths are vertex disjoint, we can define a new valuation which gives probability $1-\varepsilon$ to each edge in this path (for some small $\varepsilon$). This gives us a graph preserving valuation for which the probability of reaching $\fin$ from $a$ is close to $1$ and that from $b$ is close to $0$ (since $b$ reaches $\fail$ with probability close to $1$). This is a contradiction since $a\sim b$. Hence, this case cannot occur.
    \\
    Case 2: $b$ does not exit $S_i$. Since $b$ is an exit for $U$, this means that there exists a state $e\in bE\cap(S_i\backslash U)$. We claim that $e$ must have a path to another exit $f\in S_i\backslash U$ which does not visit any state in $U$. If this was not the case, it would mean that each path from $e$ to $\fin$ must contain a state in $U$. Now, we use \cref{lem:essential_reach} to get:\begin{align*}
            \ReachProbM{e}&=\sum_{u\in U}\ReachProbu{e}{M}{(S \setminus U)}{u}\cdot \ReachProbM{u}\\
                        &=\sum_{u\in U}\ReachProbu{e}{M}{(S \setminus U)}{u}\cdot \ReachProbM{e}&\text{$e$ is in the same layer as $U$}\\
                        &=\ReachProbM{e}\cdot\sum_{u\in U}\ReachProbu{e}{M}{(S \setminus U)}{u}\\
            \sum_{u\in U}\ReachProbu{e}{M}{(S \setminus U)}{u}&=1.
    \end{align*}
    This means that $e$ almost surely reaches $U$, and from \cref{lem:AS_implies_equivalence}, $e\in U$ which is a contradiction to our assumption. Thus, $b$ has a path to another state, say $f$, in the same layer which exits the layer such that no state apart from $b$ in the path is in $U$. Now, we have state disjoint paths from $a$ to $\fin$ and $f$ to $\fail$ and hence, $b$ to $\fail$. We can define a new valuation which gives probability $1-\varepsilon$ to each edge in this path (for some small $\varepsilon$). This gives us a graph preserving valuation for which the probability of reaching $\fin$ from $a$ is close to $1$ and that from $b$ is close to $0$. This is again, a contradiction to the fact that $a\sim b$.

    Since both the cases cannot occur, we can conclude that an equivalence class cannot have more than one exit.\qed
\end{proof}

Since we have shown that each equivalence class has exactly one exit, it immediately follows that each state in the equivalence class must almost surely reach the unique exit.

\subsubsection{Second proof}

Recall that for a state $s\in S$, $\Tilde{s}$ denotes the NWR equivalence class of $s$.
\begin{definition}[Equivalence partition]
    An \text{equivalence partition} of a Markov chain $\mathcal{M}$ is a partition $\set{S_0,\dots,S_m}$ such that the following hold:
    \begin{itemize}
        \item $S_0=\set{\fail}$, $S_m=\set{\fin}$ and for any $u\in S$, $\Tilde{u}\subseteq S_i$ for some $0<i<m$;
        \item Any $u\in S_j$ either has edges to some $S_i$ and $S_k$ with $i<j<k$ or it only has edges inside $S_j$;
        \item $S_i$ has at least one vertex for all $0\leq i\leq m$
    \end{itemize}
\end{definition}

The individual elements of such a partition are called \textit{layers} where $S_0$ is the bottom layer, $S_m$ is the top layer (similar to the partitioning in the proof of \cref{thm:w_to_b_trivparam}). Each layer consists of a (non-empty) union of equivalence classes.

We denote the set of all vertices ``above'' the layer $S_i$ as $S_{i+}=S_{i+1}\cup\dots\cup S_m$ and the set of all vertices ``below'' $S_i$ as $S_{i-}=S_0\dots S_{i-1}$.

\begin{definition}[Exits]
    For any partition of states, a state $s$ \text{exits} the partition if there is at least one edge from that vertex to another partition.
\end{definition}

Note that each layer (except the top and bottom layers with $\fin$ and $\fail$ respectively) has at least one exit since each state can reach $\fin$.

The following lemma extends \cref{lem:leave-set} to get lower bounds on the states in a Markov chain as long as they have a path to at least one exit.

\begin{lemma}
\label{lem:bound_exits}
Consider a MC $\mathcal{M}$ and let $U\subseteq S\backslash T$,  $\val\in\GraphPresVals$, and $V$ denote the set of states $s$ such that $sE \setminus U\neq\emptyset$. If $|V|\neq 0$ and for all $s\in V$ we have $x\leq\ReachProbM{s} \leq y$, then for all $u\in U$, $x\leq\ReachProbM{u} \leq y$.
\end{lemma}
\begin{proof}
    
    The proof of the upper bound is the same as the proof of \cref{lem:leave-set}. We prove the lower bound by contradiction. Suppose there is a state $u\in U$ such that $x_u=\ReachProbM{u}<x$. Let $u$ be the vertex with the least value in $U$. Since $uE=2$, the value of $u$ must be a convex combination of its two successors, say $u'$ and $u''$. But, due to the assumption that $u$ had the least value in $U$ and $u\not\in V$, we get that $u',u''$ must have the same value as $u$, and hence, both of these are also not in $V$. We continue this for $u'$ and $u''$. This procedure must stop since there are only finitely many states. Call this set of states that are reachable by $u$ as $S_u$. The reward value of each state in $S_u$ is equal to $x_u$. Since we have the assumption that each state in $U$ has a path to $\fin$, we have $\fin\in S_u$ which is a contradiction because $\ReachProbM{\fin}=1>x_u$.\qed
\end{proof}

\begin{lemma}\label{lem:partition_to_valuation}
    For any equivalence partition $\set{S_0,\dots,S_m}$, there exists a valuation $\val\in\GraphPresVals$ such that for any $s_i\in S_i$ and $s_j\in S_j$ with $i<j$, we have $\ReachProbM{s_i}<\ReachProbM{s_j}$.
\end{lemma}
\begin{proof}
    Let $0<\varepsilon_{m-1}<\dots<\varepsilon_1<1$ be small enough. Let $\set{E_1,\dots,\E_{m-1}}$ be the set of all exits in the sets $S_1,\dots,S_{m-1}$ respectively. We define the valuation $val$ as follows:
    \[
    \val(s,s')=\begin{cases}
        0.5 & \text{if $s,s'\in S_i$ for some $0<i<m$};\\
        1-\varepsilon_i & \text{if $s\in S_i$ and $s'\in S_{i+}$ for some $0<i<m$};\\
        \varepsilon_i & \text{if $s\in S_i$ and $s'\in S_{i-}$ for some $0<i<m$}.\\
    \end{cases}
    \]
    That is, all the edges within the same layer get probability $0.5$ assigned to them, the edges exiting the layer $S_i$ and going ``up'' get probability $1-\varepsilon_i$ and those going ``down'' get probability $\varepsilon$. Now, we show that assigning the probabilities in such a way gives us some interesting upper and lower bounds on the vertices in a layer.
    \begin{proposition}
        Let $\val$ be defined as above. Then for any $s\in S_i$ for some $0<i<m$ \[(1-\varepsilon_{m-1})\cdot\dots\cdot(1-\varepsilon_i)\leq\ReachProbM{s}\leq1-\varepsilon_1\cdot\dots\cdot\varepsilon_i.\]
    \end{proposition}
    \begin{proof}
        Both bounds are proved by induction. First, we will prove the lower bound by inducting starting from the top layer $S_{m-1}$. Any exit in this layer has an edge to $\fin$ with probability $1-\varepsilon_{m-1}$ (and the other edge might be to $\fail$). Thus, the bound for all vertices in $S_{m-1}$ holds by \cref{lem:bound_exits}. Now, assume that the bound $(1-\varepsilon_{m-1})\cdot\dots\cdot(1-\varepsilon_j)\leq\ReachProbM{s'}$ holds for any $s'\in S_j$ with $j>i$. Let $s$ be an exit in $S_i$. We have \begin{align*}
        \ReachProbM{s}&\geq (1-\varepsilon_i)\min_{j>i}\set{\ReachProbM{s'}|s'\in S_j}\\
        &\geq (1-\varepsilon_i)\min_{j>i}\set{(1-\varepsilon_{m-1})\cdot\dots\cdot(1-\varepsilon_j)}\\
        &= (1-\varepsilon_i)(1-\varepsilon_{m-1})\cdot\dots\cdot(1-\varepsilon_{i+1}).
        \end{align*}
        Thus, the lower bound holds for all exits in the layer and by \cref{lem:bound_exits}, it holds for all vertices in the layer.

        Now, to show the upper bound, we induct starting from the lowest layer $S_1$. Any exit in this layer goes to $\fail$ with probability $\varepsilon_1$ and thus its value can be at most $1-\varepsilon_1$. By \cref{lem:bound_exits} this bound holds for all vertices in $S_1$. Now, assume the bound $\ReachProbM{s'}\leq1-\varepsilon_1\cdot\dots\cdot\varepsilon_j$ holds for any $s'\in S_j$ with $j<i$. Let $s$ be an exit in $S_i$. In the best case scenario, it can have an edge to $\fin$ with probability $1-\varepsilon_i$ and the ``best'' vertex in $S_{i-}$ with probability $\epsilon_i$. We have \begin{align*}
            \ReachProbM{s}&\leq 1-\varepsilon_i+\varepsilon_i\cdot\max_{j<i}\set{\ReachProbM{s'}|s'\in S_j}\\
            &\leq 1-\varepsilon_i+\varepsilon_i\cdot\max_{j<i}\set{1-\varepsilon_1\cdot\dots\cdot\varepsilon_j}\\
            &=1-\varepsilon_i+\varepsilon_i\cdot(1-\varepsilon_1\cdot\dots\cdot\varepsilon_{i-1})\\
            &=1-\varepsilon_1\cdot\dots\cdot\varepsilon_i.
        \end{align*}
        Thus, the upper bound for all the exists and hence it also holds for all states in $S_i$ by \cref{lem:bound_exits}.\qed
    \end{proof}

    Back to the proof of the Lemma, it suffices to show the existence of $0<\varepsilon_{m-1}<\dots<\varepsilon_1<1$ such that the lower bound of any layer $S_i$ is greater than the upper bound of the layer $S_{i-1}$. More precisely, we want to ensure that $(1-\varepsilon_{m-1})\cdot\dots\cdot(1-\varepsilon_i)>1-\varepsilon_1\cdot\dots\cdot\varepsilon_{i-1}$ for any $1<i<m$. We have \begin{align*}
        (1-\varepsilon_{m-1})\cdot\dots\cdot(1-\varepsilon_i)&>(1-\varepsilon_i)^{m-i}&\varepsilon_i>\varepsilon_j\text{ for }j>i\\
        &>(1-\varepsilon_i)^n
    \end{align*}
    and \begin{align*}
        1-\varepsilon_1\cdot\dots\cdot\varepsilon_{i-1}&<1-\varepsilon_{i-1}^{i-1}&\varepsilon_{i-1}<\varepsilon_j\text{ for }j<i\\
        &<1-\varepsilon_{i-1}^n.
    \end{align*}
    Now, we can easily choose $\varepsilon_i$ \textit{much} smaller than $\varepsilon_{i-1}$ so that $(1-\varepsilon_i)^n>1-\varepsilon_{i-1}^n$ holds. Hence, such a sequence of $0<\varepsilon_{m-1}<\dots<\varepsilon_1<1$ exists, which concludes the proof.\qed
\end{proof}

\begin{lemma}\label{lem:valuation_to_partition}
    Any valuation $\val\in\GraphPresVals$ induces an equivalence partition $\set{S_0,\dots,S_m}$ such that two states $s,s'\in S$ are in the same layer iff $\ReachProbM{s}=\ReachProbM{s'}$.
\end{lemma}
We define the partition by grouping all the vertices with the same reward into the same layer and then ordering the layers by the rewards of the states in them. That is, layers with higher rewards are above the layers with lower rewards. One can easily check that this is indeed an equivalence partition.

\begin{definition}[Finest valuation]
    A graph preserving valuation $\val$ is a finest valuation if for all $u,v\in S$, $\ReachProbM{u}=\ReachProbM{v}$ iff $u\sim v$.
\end{definition}
\begin{definition}[Finest partition]
    A finest partition is an equivalence partition such that the number of layers in the partition is exactly the number of equivalence classes.
\end{definition}

\begin{lemma}
    A finest valuation always exists for a Markov chain.
\end{lemma}
\begin{proof}
    Suppose, that a finest valuation does not exist. Let $\val$ be a valuation such that there does not exist any other valuation such that the induced equivalence partition (\cref{lem:valuation_to_partition}) has more layers than the one induced by $\val$. Let $\set{S_0,\dots,S_m}$ be the equivalence partition induced by $\val$. Since $\val$ is not a finest valuation, there exists a layer, say $S_i$, with at least two equivalence classes. Now, we observe that at least two of these equivalence classes must have states that exit $S_i$ (otherwise, if there was only one equivalence class with exits, all other equivalence classes on that layer would almost surely reach this equivalence class, and by \cref{lem:AS_implies_equivalence}, all of these equivalence classes would be equivalent to the one with the exits, which is not possible). We now move one of these equivalence classes with exits to a new layer $S_{i'}$ which is placed in between $S_i$ and $S_{i+1}$. We call the layer obtained by removing $S_{i'}$ as $S_{i''}$. Note that each exit has one edge in an equivalence partition has an going up and one edge going down. This along with the fact that any state can almost surely reach the set of all exits, ensures that if a state $u$ is on a layer above $v$ in some equivalence partition, then there are state-disjoint paths from $u$ to $\fin$ and $v$ to $\fail$. Consider the new partition $\set{S_0,\dots,S_{i''},S_{i'},S_{i+1},\dots,S_m}$ (which is not necessarily an equivalence partition). Let $s$ be an exit in $S_{i'}$ which was also an exit in $S_i$ in the original partition. That is, $s$ has an edge to $S_{i'+}$ and one to $S_{i''-}$. We claim that $s$ is the unique vertex with outgoing edges from $S_{i'}$ (i.e. its equivalence class). Suppose there is another vertex $s'$ that exits $S_{i'}$, then it must have an edge to $S_{i''}$ or lower. Thus, there are disjoint paths from $s$ to $\fin$ and $s'$ to fail. By making the probabilities on these paths close to $1$, we can make the reward of $s$ close to $1$ and $s'$ close to $0$. This contradicts the fact that both of them are NWR equivalent. Thus, $S_{i'}$ is a ``valid'' layer of an equivalence partition. We can use a similar argument to show that any equivalence class in $S_i$ which has an exit from $S_i$ must have a unique vertex with edges outside the equivalence class. But there might be states in the original $S_i$ that had an edge to a state in $S_{i'}$. We will now make layers iteratively such that all the layers except $S_i''$ are still valid layers in each step. Let $S_{i1}\subset S_{i''}$ be one of the equivalence classes which had an edge to $S_{i'}$. We move $S_{i1}$ between $S_{i''}$ and $S_{i'}$. Using the argument from above, this state going up to $S_{i'}$ and down to $S_{i''}$ is the unique exit of this layer. We continue doing this until there are no more new exits (apart from the exits originally in $S_i$) in $S_{i''}$. This gives us an equivalence partition which is finer than the original one. Now, using \cref{lem:partition_to_valuation}, we can get a valuation $\val'$ finer than $\val$ which is a contradiction to the assumption.\qed
\end{proof}

Now, we use this finest valuation to get a finest partition and use this to show that each equivalence class has a unique state with edges going out of the equivalence class.
\begin{lemma}\label{lem:unique_exit}
    For any equivalence class $U\subset S\backslash T$, there exists a unique state $s\in S$ such that $sE\not\subseteq U$.
\end{lemma}
\begin{proof}
    Let $\val$ be a finest valuation. Let $\set{S_0,\dots,S_n}$ be a finest partition. Thus, each equivalence class has its own layer. Let $U=S_i$ for some $0<i<n$. Assume that $S_i$ has two exits $s$ and $s'$. We can now obtain state disjoint paths from $s$ to $\fin$ (by only going up) and from $s'$ to $\fail$ (by going down). We can now make the probabilities on the edges on these paths close to 1 so that the reward of $s$ is close to $1$ and that of $s'$ is close to $0$. This is a contradiction to the fact that they were in the same equivalence class.\qed
\end{proof}

Now, using \cref{lem:unique_exit} and \cref{lem:bound_exits}, we get our main result that all vertices of an equivalence class in a Markov chain can almost surely reach the unique exit.

One can easily give a polynomial time algorithm for this by checking almost-sure reachability between pairs of vertices.

\begin{theorem}\label{thm:MC_equivalence_ptime}
    Computing and collapsing the NWR equivalence classes of a Markov chain can be done in polynomial time.
\end{theorem}

\section{Proof of \cref{lem:underapprox_init}}

We first state the following lemma which says that the never-worse relation is transitive:
\begin{lemma}\label{lem:transitive}
    For any wpMDP $\mathcal{M}=(S,A,X,\delta,T,\rho)$, for any $P,Q,R\subseteq V$, if $P\trianglelefteq Q$ and $Q\trianglelefteq R$, then $P\trianglelefteq R$.
\end{lemma}
The reader can easily verify that the above lemma holds from the definition of the NWR.

Now, for the proof of \cref{lem:underapprox_init}, to show that
\eqref{eqn:underapprox} holds, due to \cref{lem:transitive}, it is enough to
show that for any $U,W\in N$, if $U\rightarrow W$ via an edge, then
$U\trianglelefteq W$. We show this for all four types of edges:
\begin{enumerate}
    \item We know that $\fail\trianglelefteq W$ for any $W\subseteq V$. Thus, an edge from $\fail$ to any node satisfies the NWR.
    \item We know that $W\trianglelefteq \fin$ for any $W\subseteq V$. Thus, an edge from any node to $\fin$ satisfies the NWR.
    \item We know, from \cref{lem:val_compare_child}, that for any valuation $\val\in\GraphPresVals$ and vertex $v\in V$, $\RewardOptVal(v)\leq\max_{u\in vE}\RewardOptVal(u)$. Hence, an edge from $v$ to $vE$ would satisfy the NWR.
    \item Again, from \cref{lem:val_compare_child}, we know that for a valuation $\val\in\GraphPresVals$ and a state $v\in S$, $\RewardOptVal(v)=\max_{u\in vE}\RewardOptVal{u}$. Thus, $v$ is never-worse than each state-action pair in $vE$ and hence, the edge satisfies the NWR.
\end{enumerate}
This concludes the proof.\qed

\section{Proof of \cref{lem:underapprox_add_edge}}

From \cref{lem:transitive}, if $\mathcal{U'}$ satisfied invariant \eqref{eqn:underapprox}, to show that $\mathcal{U}$ satisfies invariant \eqref{eqn:underapprox} too, it is enough to show that all the new edges that were added satisfy the NWR. The edges added from (to) the subsets (supersets) of $U$ and $W$ satisfy the NWR because for any $P\subseteq P'\subseteq V$, it is true that $P\trianglelefteq P'$ by definition. It is given that $U\trianglelefteq W$, thus adding the edge $U\rightarrow W$ (if there is no path from $U$ to $W$) should also satisfy the NWR.

\section{Proof of \cref{thm:mec_edge_removal}}

\begin{proof}[Proof of Theorem~\ref{thm:mec_edge_removal}]For simplicity, we rename the state-action pair $(s,a)$ as $x$. We will split the proof into two cases: one which considers the possibility of $s$ getting disconnected from $\fin$ after removing the edge and the other in which $s$ still has a path to $\fin$ after removing the edge.

\textbf{Case 1}, i.e. $s$ has no path to $\fin$ after removing the edge $(s,x)$:
This would imply that no vertex in $sE\backslash\Set{x}$ has a path to $\fin$
after removing the edge. Let $u\in sE\backslash \Set{x}$. Thus, for any
$\sigma$, we have
$\ReachProb{v}{\mathcal{M}_{\sigma}^\val}{\fin}=\ReachProb{v}{\mathcal{M}_{\sigma}^\val}{s}\cdot\ReachProb{s}{\mathcal{M}_{\sigma}^\val}{\fin}$
for all $v\in uE$ (from \cref{lem:essential_reach}). If we have
$\ReachProb{v}{\mathcal{M}_{\sigma}^\val}{s}=1$ for all $v\in uE$, then it
would mean that $v$ and $s$ are in the same MEC, and hence, by our assumption
that $\mathcal{M}$ is isomorphic with $\mmec$, we deduce that $v=s$ and hence
$uE=\Set{s}$. But this cannot be the case since we would not have the edge
from $s$ to $u$ if $uE=\Set{s}$. This means that there must be at least one
vertex $v\in uE$ for which $\ReachProb{v}{\mathcal{M}_{\sigma}^\val}{s}<1$.
Hence, we have
$\ReachProb{u}{\mathcal{M}_{\sigma}^\val}{\fin}<\ReachProb{s}{\mathcal{M}_{\sigma}^\val}{\fin}$.
Since $u$ was chosen arbitrarily from $sE\backslash \Set{x}$, and because $\ReachProb{u}{\mathcal{M}_{\sigma}^\val}{\fin}<\ReachProb{s}{\mathcal{M}_{\sigma}^\val}{\fin}$, we get that
\begin{align*}
    \RewardOptValM(s)&=\max\Set{\RewardOptValM(u)\mid u\in sE}\\
    &=\max\left\{\left(\max_{u \in sE\backslash
    \Set{x}}\RewardOptValM(u)\right)\cup\RewardOptValM(x)\right\}\\
    &=\RewardOptValM(x).
\end{align*}
Now, we get the contradiction that $\RewardOptValM(x)=\RewardOptValM(s)>\RewardOptValM(u)$ for $u\in sE\backslash \Set{x}$ since our assumption was that $x\trianglelefteq(sE \setminus \Set{x})$. Hence, Case 1 cannot happen.

\textbf{Case 2}, i.e. the connectivity of $s$ to $\fin$ does not change after removing the $s\rightarrow x$ edge:
For any valuation $\val\in\GraphPresVals$, using \cref{lem:bellman_eq}, we know that for the following system of linear equations
\[\text{For all }s\in S,\ y(s) = \begin{cases}
    1       & \text{if}\; s=\fin; \\
    0       & \text{if}\; s \in \mathbf{Z}_\fin; \\
    \max_{v \in sE}\sum_{s'\in vE}\delta(v,s')[\val]\cdot y(s') & \text{otherwise}
\end{cases}\]
has a unique solution \[y(u)=\max_\sigma\ReachProb{u}{\mathcal{M}_{\sigma}^\val}{\fin}.\]
Since we have $x\trianglelefteq(sE \setminus \Set{x})$, we assert that the following system of equations has the same unique solution:
\[\text{For all }u\in S,\ y(u) = \begin{cases}
    1       & \text{if}\; u=\fin; \\
    0       & \text{if}\; u \in \mathbf{Z}_\fin; \\
    \max_{v \in uE\setminus\Set{x}}\sum_{s'\in vE}\delta(v,s')\cdot y(s') & \text{if }u=s;\\
    \max_{v \in uE}\sum_{s'\in vE}\mu_v(s')\cdot y(s') & \text{otherwise}.
\end{cases}\]
We have already shown that \textbf{Case 1} cannot happen. That is, removing the edge from $s$ to $x$ does not change the set $\mathbf{Z}_\fin$. Thus, the solution of this modified system has the unique solution: \[y(u)=\max_\sigma\ReachProb{u}{\mathcal{N}_{\sigma}^\val}{\fin}.\]
Thus, the values of the vertices are preserved even after we remove the edge
$s\rightarrow x$.\qed
\end{proof}

\end{document}

%% file: images_tikz/plot_tikz_43.tikz
\begin{tikzpicture}[scale=.65]

\definecolor{darkgray176}{RGB}{176,176,176}
\definecolor{darkorange25512714}{RGB}{255,127,14}
\definecolor{lightgray204}{RGB}{204,204,204}
\definecolor{steelblue31119180}{RGB}{31,119,180}

\begin{axis}[
legend cell align={left},
legend style={fill opacity=0.8, draw opacity=1, text opacity=1, draw=lightgray204},
align =center,
tick align=outside,
tick pos=left,
title={\texttt{brp} (N=64,MAX=5) \\ \texttt{"no\_succ\_trans"}},
x grid style={darkgray176},
xlabel={Stage},
xmin=-0.15, xmax=3.15,
xtick style={color=black},
xtick={0,1,2,3},
xtick={0,1,2,3},
xtick={0,1,2,3},
xticklabels={Start,Iter0,Iter1,Iter2},
xticklabels={Start,Iter0,Iter1,Iter2},
xticklabels={Start,Iter0,Iter1,Iter2},
y grid style={darkgray176},
ylabel={Count},
ymin=0, ymax=4609.45,
ytick style={color=black}
]
\addplot [semithick, steelblue31119180]
table {%
0 4472
1 1725
2 1725
3 1725
};
\addlegendentry{States}
\end{axis}

\end{tikzpicture}

%% file: images_tikz/plot_tikz_4.tikz
\begin{tikzpicture}[scale=.65]

\definecolor{darkgray176}{RGB}{176,176,176}
\definecolor{darkorange25512714}{RGB}{255,127,14}
\definecolor{lightgray204}{RGB}{204,204,204}
\definecolor{steelblue31119180}{RGB}{31,119,180}

\begin{axis}[
legend cell align={left},
legend style={fill opacity=0.8, draw opacity=1, text opacity=1, draw=lightgray204},
align =center,
tick align=outside,
tick pos=left,
title={\texttt{consensus} (N=2,K=16) \\\texttt{"fin\_and\_all\_1"}},
x grid style={darkgray176},
xlabel={Stage},
xmin=-0.15, xmax=3.15,
xtick style={color=black},
xtick={0,1,2,3},
xtick={0,1,2,3},
xtick={0,1,2,3},
xticklabels={Start,Iter0,Iter1,Iter2},
xticklabels={Start,Iter0,Iter1,Iter2},
xticklabels={Start,Iter0,Iter1,Iter2},
y grid style={darkgray176},
ylabel={Count},
ymin=0, ymax=2589.15,
ytick style={color=black}
]
\addplot [semithick, steelblue31119180]
table {%
0 1517
1 1137
2 1137
3 1137
};
\addlegendentry{States}
\addplot [semithick, darkorange25512714]
table {%
0 2520
1 2140
2 2140
3 2140
};
\addlegendentry{Choices}
\end{axis}

\end{tikzpicture}

%% file: images_tikz/plot_tikz_88.tikz
\begin{tikzpicture}[scale=.65]

\definecolor{darkgray176}{RGB}{176,176,176}
\definecolor{darkorange25512714}{RGB}{255,127,14}
\definecolor{lightgray204}{RGB}{204,204,204}
\definecolor{steelblue31119180}{RGB}{31,119,180}

\begin{axis}[
legend cell align={left},
legend style={fill opacity=0.8, draw opacity=1, text opacity=1, draw=lightgray204},
align =center,
tick align=outside,
tick pos=left,
title={\texttt{zeroconf} (N=1000,K=2,reset=True)\\ \texttt{"correct\_max"}},
x grid style={darkgray176},
xlabel={Stage},
xmin=-0.4, xmax=8.4,
xtick style={color=black},
xtick={0,1,2,3,4,5,6,7,8},
xtick={0,1,2,3,4,5,6,7,8},
xtick={0,1,2,3,4,5,6,7,8},
xticklabels={Start,It0,It1,It2,It3,It4,It5,It6,It7},
xticklabels={Start,It0,It1,It2,It3,It4,It5,It6,It7},
xticklabels={Start,It0,It1,It2,It3,It4,It5,It6,It7},
y grid style={darkgray176},
ylabel={Count},
ymin=0, ymax=499.15,
ytick style={color=black}
]
\addplot [semithick, steelblue31119180]
table {%
0 388
1 130
2 130
3 127
4 124
5 121
6 118
7 118
8 118
};
\addlegendentry{States}
\addplot [semithick, darkorange25512714]
table {%
0 481
1 174
2 171
3 165
4 159
5 153
6 150
7 150
8 150
};
\addlegendentry{Choices}
\end{axis}

\end{tikzpicture}

%% file: images_tikz/plot_tikz_9.tikz
\begin{tikzpicture}[scale=.65]

\definecolor{darkgray176}{RGB}{176,176,176}
\definecolor{darkorange25512714}{RGB}{255,127,14}
\definecolor{lightgray204}{RGB}{204,204,204}
\definecolor{steelblue31119180}{RGB}{31,119,180}

\begin{axis}[
legend cell align={left},
legend style={fill opacity=0.8, draw opacity=1, text opacity=1, draw=lightgray204},
align =center,
tick align=outside,
tick pos=left,
title={\texttt{zeroconf\_dl} (N=1000,K=1,res=true,dl=50) \\\texttt{"deadline\_min"}},
x grid style={darkgray176},
xlabel={Stage},
xmin=-0.15, xmax=3.15,
xtick style={color=black},
xtick={0,1,2,3},
xtick={0,1,2,3},
xtick={0,1,2,3},
xticklabels={Start,Iter0,Iter1,Iter2},
xticklabels={Start,Iter0,Iter1,Iter2},
xticklabels={Start,Iter0,Iter1,Iter2},
y grid style={darkgray176},
ylabel={Count},
ymin=0, ymax=12296.9,
ytick style={color=black}
]
\addplot [semithick, steelblue31119180]
table {%
0 9579
1 4613
2 4613
3 4613
};
\addlegendentry{States}
\addplot [semithick, darkorange25512714]
table {%
0 11931
1 6925
2 6925
3 6925
};
\addlegendentry{Choices}
\end{axis}

\end{tikzpicture}